\documentclass[12pt]{article}
\usepackage{amsmath}
\usepackage{amsmath,amssymb,amsthm,graphicx,subfigure}
\usepackage[title]{appendix}

\usepackage{graphicx}
\usepackage{enumerate}
\usepackage{natbib}
\usepackage{url} 

\usepackage{multirow,threeparttable,float}
\usepackage{xcolor}
\definecolor{darkred}{HTML}{8b0000}
\definecolor{darkblue}{HTML}{00008b}
\definecolor{darkgreen}{HTML}{006400}

\usepackage{algorithm}
\usepackage{algorithmicx}

\usepackage{algcompatible}

\newtheorem{theorem}{Theorem}
\newtheorem{prop}[theorem]{Proposition}

\def\bsq#1{\lq{#1}\rq}

\def\N{\mathcal{N}}
\def\T{\mathcal{T}}
\def\R{\mathbb{R}}
\def\he{\hat{\varepsilon}}
\def\est{\hat{\theta}}
\def\E{\mathbb{E}}

\def\D{\mathcal{F}}
\def\IG{\mathrm{IG}}
\def\prior{\mathrm{Priors}_{t=0}}

\DeclareMathOperator*{\argmin}{arg\,min}

\newcommand{\dJdl}[1]{\frac{ \partial J_{#1} }{ \partial \lambda} }

\newcommand{\dTdl}[1]{\frac{ \partial \est_{#1}}{ \partial \lambda} }

\newcommand{\dCdl}[1]{ \frac{ \partial  C_{#1}}{ \partial \lambda} }

\newcommand{\dSdl}[1]{ \frac{ \partial S_{#1} }{ \partial \lambda} }
\newcommand{\dQdl}[1]{ \frac{ \partial Q_{#1} }{ \partial \lambda} }

\newcommand{\dAdl}[1]{ \frac{ \partial  A_{#1}}{ \partial \lambda} }

\newcommand{\blind}{0}

\begin{document}

\def\spacingset#1{\renewcommand{\baselinestretch}%
{#1}\small\normalsize} \spacingset{1}


\if0\blind
{
  \title{\bf Adaptive Dynamic Model Averaging with an Application to House Price Forecasting}
  \author{Alisa Yusupova
  \hspace{.2cm}\\ Department of Management Science, Lancaster University\\
     \\
    Nicos G. Pavlidis \\
    Department of Management Science, Lancaster University\\
     \\
    Efthymios G. Pavlidis \\
    Department of Economics, Lancaster University }
  \maketitle
} \fi

\if1\blind
{
  \bigskip
  \bigskip
  \bigskip
  \begin{center}
    {\LARGE\bf Adaptive Dynamic Model Averaging with an Application to House Price Forecasting}
\end{center}
  \medskip
} \fi

\bigskip

\begin{abstract}

Dynamic model averaging (DMA) combines the forecasts of a large number of
dynamic linear models (DLMs) to predict the future value of a time series. The
performance of DMA critically depends on the appropriate choice of two
forgetting factors. The first of these controls the speed of adaptation of the
coefficient vector of each DLM, while the second enables time variation in the
model averaging stage. In this paper we develop a novel, adaptive dynamic model
averaging (ADMA) methodology. The proposed methodology employs a stochastic
optimisation algorithm that sequentially updates the forgetting factor of each
DLM, and uses a state-of-the-art non-parametric model combination algorithm
from the prediction with expert advice literature, which offers finite-time
performance guarantees. An empirical application to quarterly UK house price
data suggests that ADMA produces more accurate forecasts than the benchmark
autoregressive model, as well as competing DMA specifications.

\end{abstract}

\noindent%
\vfill

\newpage
\spacingset{1.45} 

\section{Introduction}\label{sec:intro}

A growing empirical literature 
provides strong evidence in favour of  structural instability in many macroeconomic relations
\citep{StockW1996,KoopP2007,NgW2013}.
Structural instability is crucial because, if left
unaccounted for, it can have detrimental consequences for statistical inference
and forecasting \citep{ClementsH1998,PesaranPT2006,GiacominiR2009,rossi2013advances}.

Dynamic Model Averaging (DMA) is an econometric methodology that can
accommodate time variations in both model parameters and model
specification~\citep{RafteryDMA2010}. 
This methodology has gained increasing popularity  in recent years for predicting various
economic variables, such as
inflation~\citep{KoopKorobilis_2012,CataniaNonejad_2018}, carbon
prices~\citep{KoopTole_2013}, exchange rates~\citep{ByrneKR2018}, equity
returns~\citep{DanglH2012}, and property price growth~\citep{BorkMoller_2015}.
DMA creates a Dynamic Linear Model (DLM) for every possible subset of  predictors and combines the forecasts of these models using weights that adapt over time~\citep{RafteryDMA2010}. To adapt to changes
in the data generating distribution, DMA involves two parameters, called
\textit{forgetting factors}.
The first forgetting factor is part of the DLM
formulation, while the second is involved in the model averaging phase. 
These parameters allow a continuous trade-off between estimation
in a static environment, and re-initialising the estimation process by
discarding all past information, which is appropriate
after a structural break.
It is therefore not surprising that the choice of forgetting factors
is critical to the forecast performance of DMA.

Initial work by \citet{RafteryDMA2010} and \citet{KoopKorobilis_2012}
considered both forgetting factors to be user-defined and constant.  However,
in general, the type of structural change in economic relationships is unknown
and may vary considerably over time \citep{ChenH2012}. Structural breaks due to
changes in regulatory conditions, in the behaviour of consumers and firms or in
the preferences of policy makers constitute prominent examples where periods
during which the data generating process is static are interrupted by episodes
of abrupt change \citep{PesaranPT2006,KAPETANIOS2010}. In this
setting, and more generally whenever the speed or type of change in the data
generating process is not constant, there does not exist a single choice of
forgetting factors that is optimal for the entire length of a time series.

The DMA formulation of~\cite{DanglH2012}, adopted by a number of more recent
works~\citep{CataniaNonejad_2018,ByrneKR2018}
involves no forgetting in the model averaging stage, and treats the choice of
the DLM forgetting factor as an additional dimension of model uncertainty.
In particular, the user specifies a grid of forgetting factor values, and the
posterior distribution of this parameter is updated at every time step by
marginalising over all DLM specifications.
The Bayesian approach of~\cite{DanglH2012} effectively assumes that the
appropriate choice of the DLM forgetting factor is constant over time, and
identical across models.
\cite{McCormick2012} propose a similar approach that involves a grid of values
for both forgetting factors. They propose to use the forgetting factors that maximise the
predictive likelihood for each DLM specification to avoid the computational
cost of Bayesian updating. 

In this paper, we develop an adaptive dynamic model averaging (ADMA)
methodology that consists of two components. The first involves the use of
stochastic optimisation to identify the forgetting factor that minimises the
expected one-step-ahead squared forecast error of each DLM. This leads to a fully
online and data-driven algorithm which we call Adaptive Forgetting DLM
(AF-DLM).  As we show in the experimental results section AF-DLM is effective
under different types of change in the data generating process, including cases
in which the speed or type of change is variable over time. AF-DLM is also
computationally less expensive than previous approaches since it does not
involve a grid of forgetting factor values.

The second component of ADMA deals with model averaging. We show that the speed
with which DMA weights respond to more recent observations is determined not
only by the choice of the corresponding forgetting factor, but also by the
mechanism that prevents underflow (weights becoming equal to machine zero).
Therefore any approach to control the adaptability of DMA weights by tuning
only the second forgetting factor is inherently limited. 
We propose to replace the current model averaging approach with the ConfHedge
model combination algorithm from the field of machine learning known as
prediction with expert advice~\citep{VyuginTrunov_2019}. In addition to being
parameter-free, ConfHedge is currently the only model combination algorithm
whose one-step-ahead squared forecast error over a finite number of time steps
is within a known bound of the one-step-ahead squared forecast error of the
optimal sequence of forecasting models.

To assess the effectiveness of the proposed methodology, we provide an in-depth
empirical evaluation of ADMA on the task of forecasting UK house prices. The
motivation for this application is twofold. First, although the latest
boom-bust episode in real estate markets and its decisive role in the Great
Recession has generated a vast interest in the behavior of international
housing markets, the academic literature on house price forecastability is
relatively small (especially when compared to the literature on other assets
such as stock prices and exchange rates), and  mainly concentrates on the US
market \citep[see, e.g.,][]{RapachStrauss_2009,GHYSELS20135,BorkMoller_2015}.
In the UK, similarly to the US, housing activities account for a large fraction
of  GDP and of households' expenditures, real estate property comprises the
largest component of private wealth (excluding private pensions), and mortgage
debt constitutes the main liability of households~\citep{ONS_2018}. Thus,
accurate forecasts of UK real estate prices are crucial for private investors
and policy makers. Second, recent empirical evidence suggests that the
relationship between real estate valuations and conditioning macro and
financial variables displays time-varying patterns
\citep{Aizenman2014,Anundsen2015,Paul2018}. This makes housing markets an ideal
setting for the application of dynamic econometric models.

In summary, the results of our empirical application suggest that ADMA offers
significant forecasting gains relative to a linear autoregressive (AR)
benchmark, as well as a battery of competing dynamic and static forecasting
models. They also indicate that the \textit{best} house predictors
substantially differ over time and across regions. In-sample stability tests
also support the conclusion that the data generating process of regional UK
house prices is subject to structural instability.

The remaining paper is organised as follows. Section~\ref{sec:meth} presents
the ADMA methodology.
In Section~\ref{sec:Exp} we assess the proposed AF-DLM on simulated time series
exhibiting different types of dynamics. Section~\ref{sec:House} is devoted to
the comparative evaluation of ADMA against alternative forecasting models on
the task of predicting UK regional house prices. The paper ends with concluding
remarks in Section~\ref{sec:Conc}.

\section{Methodology}\label{sec:meth}

This section is divided into three parts. The first part provides a brief
outline of the DMA methodology. The second part deals with the development of
the stochastic optimisation approach to sequentially adapt the forgetting
factor in a single DLM. The last part discusses limitations of existing model
combination methods and presents the ConfHedge algorithm.

\subsection{Dynamic Model Averaging}

For a set of~$D$ covariates, DMA creates a DLM for each possible subset (excluding the empty set), giving rise to~$K=2^D-1$ models, $M_1,\ldots,M_K$. Model $M_k$ is defined by, 
\begin{align}
\theta^{(k)}_{t+1} & = \theta^{(k)}_{t} + \omega^{(k)}_{t+1}, & \omega^{(k)}_{t+1} \sim \N\left(0, W^{(k)}_{t+1} \right), \label{eq:state}\\
y_{t+1} & = x_{t+1}^{(k)\top} \theta^{(k)}_{t+1} + \varepsilon^{(k)}_{t+1}, & \varepsilon^{(k)}_{t+1} \sim \N \left(0, V^{(k)}_{t+1}\right), \label{eq:meas}
\end{align}
where $\theta^{(k)}_t \in \R^d$ denotes the coefficient vector and $x^{(k)}_t \in \R^d$ is the covariate vector (including a constant) of $M_k$ at time~$t$. Eq.~(\ref{eq:state}), known as the state-transition equation, determines the dynamics of the unobserved coefficient vector. Eq.~(\ref{eq:meas}), called the observation or measurement equation, links the response, $y_{t+1}$, to the coefficients and the covariates.
The errors, or noise terms, $\omega^{(k)}_{t+1}$ and $\varepsilon^{(k)}_{t+1}$ in
Eqs.~(\ref{eq:state}) and~(\ref{eq:meas}), respectively, are assumed to be
independent normally distributed random variables. The state transition
covariance matrix, $W^{(k)}_{t+1}$, and the observational variance, $V^{(k)}_{t+1}$, are typically unknown.

The DMA forecast, $\hat{y}_{t+1}$, is obtained through
a convex combination of the forecasts of the $K$ DLMs,
\begin{equation}\label{eq:DMAver}
\hat{y}_{t+1} = \sum_{k=1}^K p(M_k| \D^{(k)}_{t}) \, \hat{y}^{(k)}_{t+1},
\end{equation}
where $\hat{y}^{(k)}_{t+1}$ is the prediction by model $M_k$, and
$p(M_k|\D^{(k)}_{t})$ is the probability (weight) assigned to $M_k$ conditional on the information 
available at time $t$, $\D^{(k)}_{t}$.
The information set is defined as,
$\D^{(k)}_{t}= \{y_{t},\ldots,y_1,
x^{(k)}_{t+1},x^{(k)}_{t},\ldots,x^{(k)}_1,\prior\}$, and
contains the choice of priors, the realisations of the covariate vector and of the response up to time $t$, as well as the covariate vector at time $t+1$,
$x^{(k)}_{t+1}$, required to predict $y_{t+1}$.

\subsection{Dynamic Linear Model with Forgetting}\label{ssec:dlm}

Because we consider a single DLM with forgetting throughout this subsection, we drop the superscript $(k)$ to simplify notation. We adopt the DLM formulation proposed of \cite{DanglH2012}, in which the prior distribution for the coefficients vector, $\theta_0$, is Gaussian;
the measurement variance, $V_t=V$, is constant over time; and an inverse-gamma distributed prior
for~$V$ is used. Consequently,
\begin{align}
V | \D_0 & \sim \IG \left(\frac{1}{2}, \frac{1}{2} S_0\right),\\
\theta_0 | \D_0, V & \sim \N \left(0, g I \right), \label{eq:C0}
\end{align}
which enables a conjugate Bayesian analysis.
The model specification and the assumptions about the priors imply that, 
\begin{align}
V | \D_{t} & \sim \IG \left(\frac{n_{t}}{2}, \frac{n_{t}}{2}
S_{t}\right),
\end{align}
where $S_{t}$ is the mean of the estimate of~$V$ at time~$t$, and $n_{t}$ stands for the associated degrees of freedom. Conditional on $V$ the posterior distribution of the coefficient vector is Gaussian, 
\begin{align}
\theta_{t} | \D_{t}, V & \sim \N \left( \est_{t}, \frac{V}{S_{t}}
C_{t} \right),
\end{align} 
where $\est_{t}$ is the point estimate of $\theta_{t}$,
and~$C_{t}$ is the estimator for the conditional covariance matrix for
$\theta_{t}$.
Integrating out $V$, the posterior becomes a multivariate $t$-distribution,
\begin{align}
\theta_{t} |\D_{t} & \sim \T_{n_{t}} \left( \est_{t}, C_{t} \right).
\end{align}
The prior distribution for the coefficient vector at the next time-step is,
\begin{align}\label{eq:theta1}
\theta_{t+1} | \D_{t} \sim \T_{n_{t}} \left( \est_{t}, C_{t} + W_{t+1} \right).
\end{align}
For a generic DLM, $W_{t+1}$ can be sequentially estimated, but the associated computational cost is prohibitive for a method like DMA which uses $2^D-1$ DLMs. The distinctive feature of the approach of \cite{RafteryDMA2010} is that it avoids altogether the estimation of $W_{t+1}$ by setting,
\begin{equation}\label{eq:Wt}
W_{t+1} = \frac{1 - \lambda}{\lambda}  C_{t},
\end{equation}
where $\lambda \in (0,1]$ is the {\em forgetting factor} parameter.
This simplifies Eq.~\eqref{eq:theta1} to,
\begin{align}\label{eq:theta2}
\theta_{t+1} | \D_{t} \sim \T_{n_{t}} \left( \est_{t}, \lambda^{-1} C_{t}
\right).
\end{align}
The forgetting factor in Eq.~\eqref{eq:Wt} allows a continuous range between estimation in a static environment and completely discarding all past data which is appropriate in response to a structural break.
Specifically, setting $\lambda=1$ corresponds to $W_t=\mathbf{0}$ and thus
reduces the DLM to a static linear model. On the other hand, as $\lambda$ tends to zero $W_t$ tends to infinity. This inflates the uncertainty about $\theta_{t+1}$ (see Eq.~\eqref{eq:theta2}), which effectively re-initialises the estimation process.
Intermediate values of $\lambda$ correspond to a random walk process 
for $\theta$ in which
periods of high estimation error in the coefficients coincide
with periods of high variability, and vice versa.

The prior distribution (predictive density) of $y_{t+1}$ conditional on $\D_{t}$
is,
\begin{align}
y_{t+1}|\D_{t} & \sim \mathcal{T}_{n_{t}} \left( x_{t+1}^\top \est_{t}, \;
Q_{t+1} \right), \label{eq:pdf}\\
Q_{t+1} & = \lambda^{-1} x_{t+1}^\top C_{t} x_{t+1} + S_{t}. \label{eq:Q}
\end{align}
Once the actual value $y_{t+1}$ is observed, we can compute the forecast error, 
\begin{equation}
\he_{t+1} = y_{t+1} - x_{t+1}^\top \est_{t},
\end{equation}
and update the prior distributions of the coefficient vector and the measurement variance through Eqs~\eqref{eq:Update1}--\eqref{eq:Koop9}.
The degrees of freedom and the estimator of the observational variance are updated according to,
\begin{align}\label{eq:Update1}
n_{t+1}  & = n_{t} +1,  \\
S_{t+1} & = S_{t} + \frac{S_{t}}{n_{t+1}} \left(\frac{\he^2_{t+1}}{Q_{t+1}} - 1 \right).
\label{eq:PW_CN}
\end{align}
The point estimate and the estimator of the covariance matrix of $\theta_{t+1}$ are obtained by,
\begin{align}
\est_{t+1} & = \est_{t} + A_{t+1} \he_{t+1},\label{eq:Koop8} \\
C_{t+1} & = \lambda^{-1} C_{t} + A_{t+1} A_{t+1}^\top Q_{t+1},\label{eq:Koop9}
\end{align}
where
\begin{equation}\label{eq:KG}
A_{t+1} =  \frac{\lambda^{-1}C_{t} x_{t+1}}{Q_{t+1}},
\end{equation}
is a coefficient vector (known as the Kalman gain), which
measures the information content of $y_{t+1}$ in relation
to the precision of the estimated regression coefficient.

\subsubsection*{Adaptive Forgetting}

We are now in a position to describe the adaptive forgetting DLM (AF-DLM).
The approach we propose is motivated by the adaptive recursive least squares algorithm~\citep{Haykin_2002}.
The central idea underlying this approach is that the optimal forgetting factor
minimises the expectation of the one-step-ahead squared forecast error,
\begin{equation}\label{eq:Exp}
\lambda^\star_{t} = \argmin_{\lambda \in (0,1]} \E \left[ \frac{1}{2} \left(y_{t+1}
- x_{t+1}^\top \est_t \right)^2 \right].
\end{equation}
Since the expectation in the above equation is not available in
analytical form it is not feasible to directly optimise it.
However, the one-step-ahead squared forecast error,
\begin{equation}\label{eq:Jt}
J_{t+1} =  \frac{1}{2} \left( y_{t+1} - x_{t+1}^\top \est_t  \right)^2,
\end{equation}
is an unbiased estimator of the expectation in Eq.~\eqref{eq:Exp}.
Stochastic optimisation algorithms are designed to
optimise the expected value of a function which depends on a set of
random variables.
The most widely used stochastic optimisation algorithms involve first-order
information and are hence variations of stochastic gradient descent. The term
stochastic gradient in this context refers to the fact that the gradient
of $J_{t+1}$ (which is an unbiased estimate of the gradient of $\E_{X,Y}[J_{t+1}]$) is used.

To use stochastic gradient descent to minimise the expected one-step-ahead
squared forecast error we need an expression for the derivative of
$J_{t+1}$
with respect to~$\lambda$.
Applying the chain rule produces,
\begin{align}\label{eq:ChainRule}
\dJdl{t+1} & = \frac{\partial J_{t+1}}{\partial \est_t} \, \dTdl{t}, \nonumber
\\
& = - \he_{t+1} x_{t+1}^\top \, \dTdl{t}.
\end{align}
To obtain $\dTdl{t}$ we differentiate the update equation for $\est_{t}$,
Eq.~\eqref{eq:Koop8}, with respect to $\lambda$.
Such an approach is utilised in a number of adaptive linear
filters~\citep{Haykin_2002}, in online neural network
training~\citep{Almeida1999,Schraudolph99b,Baydin2018}, as well as in streaming
data classifiers~\citep{PavlidisTAH2011,Anagn12}.
The outcome of this differentiation is,
\begin{align}
\dTdl{t} & = \nabla_{\lambda} \left\{  \est_{t-1} + A_{t} \he_{t} \right\}
\nonumber \\
& =  (I - A_{t-1}x_t^\top) \dTdl{t-1} + \frac{\he_t}{S_{t-1}} \left(\dCdl{t}
x_{t} - A_{t-1} \dSdl{t-1} \right).
\end{align}
The derivation of all the necessary quantities to estimate $\dTdl{t}$ is lengthy and is hence provided in Appendix~\ref{app:Deriv}.

Using the gradient $\dJdl{t+1}$, AF-DLM updates the value of $\lambda$ at each time-step through
the highly influential adaptive moment estimation (ADAM) stochastic gradient
descent algorithm~\citep{Diederik2015}.
ADAM uses adaptive estimates of the first two moments of the stochastic
gradient to tune the crucial step-size parameter of the gradient
descent algorithm.
ADAM has been shown to be effective in a wide range of challenging optimisation
problems; it is straightforward to implement; and is capable of handling
non-stationary objective functions~\citep{Baydin2018}. The latter aspect is
particularly important for our purposes 
since the optimal value of~$\lambda$ is itself time-varying for time series
that exhibit structural breaks, or more generally non-constant dynamics.

\subsection{Model Averaging}\label{ssec:ma}

In this section we discuss the model averaging component of DMA.
Recall that according to Eq.~\eqref{eq:DMAver} the DMA forecast, $\hat{y}_{t+1}$, 
is a convex combination of the forecasts produced by the~$K$ DLMs.
Let $p(M_k|\D_t)$ denote the probability (weight)
of model $M_k$ conditional 
on the information set $\D_t$. The following definition of $p(M_k|\D_t)$ accommodates
all DMA variants,
\begin{align}
p(M_k|\D_t) & = \frac{ p(y_t| M_k, \D_{t-1}) \left[ p(M_k | \D_{t-1})^\alpha + c \right] }
	{\sum_{m=1}^K p(y_t| M_m, \D_{t-1}) \left[ p(M_m | \D_{t-1})^\alpha + c \right] }. \label{eq:KK_16}
\end{align}
Ignoring the constant $c$ and expanding Eq.~\eqref{eq:KK_16} gives,
\begin{align*}
p(M_k|\D_t) & = \frac{ p(M_k|\D_0)^{\alpha^t} \prod_{j=1}^t p\left(y_j|M_k, \D_{j-1}\right)^{\alpha^{t-j} }}
	{ \sum_{m=1}^K p(M_m|\D_0)^{\alpha^t} \prod_{j=1}^t p \left(y_j|M_m, \D_{j-1} \right)^{\alpha^{t-j}} }, \\
	& \propto \prod_{j=1}^t p(y_j|M_k, \D_{j-1}) ^{\alpha^{t-i}},
\end{align*}
where typically, $p(M_k|\D_0)= 1/K$.
The last equation 
highlights that $\alpha$ controls the rate at which past information is discounted.

The recommendations in the literature are to set the forgetting factor $\alpha$ either equal to one
(which corresponds to Bayesian Model Averaging), or very close to unity.  In
particular, \cite{DanglH2012} and \cite{ByrneKR2018} use $\alpha=1$, while in
the eDMA {\tt R} package of \cite{CataniaNonejad_2018} the default is
$\alpha=0.99$.  \cite{RafteryDMA2010} and \cite{KoopKorobilis_2012} recommend
using $\lambda=\alpha$.  Specifically, \cite{RafteryDMA2010} recommend
$\lambda=\alpha=0.99$, while \cite{KoopKorobilis_2012} consider two values
$\lambda,\alpha \in \{0.95, 0.99\}$.
Only \cite{RafteryDMA2010} mention the small positive constant, $c$, in
Eq.~\eqref{eq:KK_16} which is included to avoid the weight of any model
becoming equal to machine zero.
Such a constant is present in the DMA implementation
of~\cite{KoopKorobilis_2012}, but not in the {\tt R} package
eDMA~\citep{CataniaNonejad_2018}.
The example we discuss next aims to illustrate that the speed with which model
probabilities adapt in response to changes in the optimal model specification
depends critically on~$c$.

For simplicity we consider a problem involving only three models whose
coefficients at every time-step are known.
The observational variance of each model is also known and constant over time.
The response at each time-step is generated from one of the three models, but
the identity of this model is unknown.
In this setting no generality is lost if we assume that for every model, $y_t^{(k)} |
\theta_t^{(k)}, x_t \sim \N(\mu_k, V_k)$, for all $t=1,\ldots,T$, and $k=1,2,3$.
Thus we assume 
$y^{(1)}|\theta_t^{(1)},x_t \sim \N(1, 4)$, 
$y^{(2)}|\theta_t^{(2)},x_t \sim \N(0, 0.64)$, 
and $y^{(3)}|\theta_t^{(3)},x_t \sim \N(-2.5, 0.09)$.
We construct a time series of length~$T=300$ 
by sampling $y_t$ from $M_1$ for $t=1,\ldots,100$, from $M_2$ for $t=101,\ldots,200$,
and from $M_3$ for $t=201,\ldots,300$.
The three models are assigned equal prior probabilities.

\begin{figure}[!h]
\begin{center}
\includegraphics[width=0.9\linewidth]{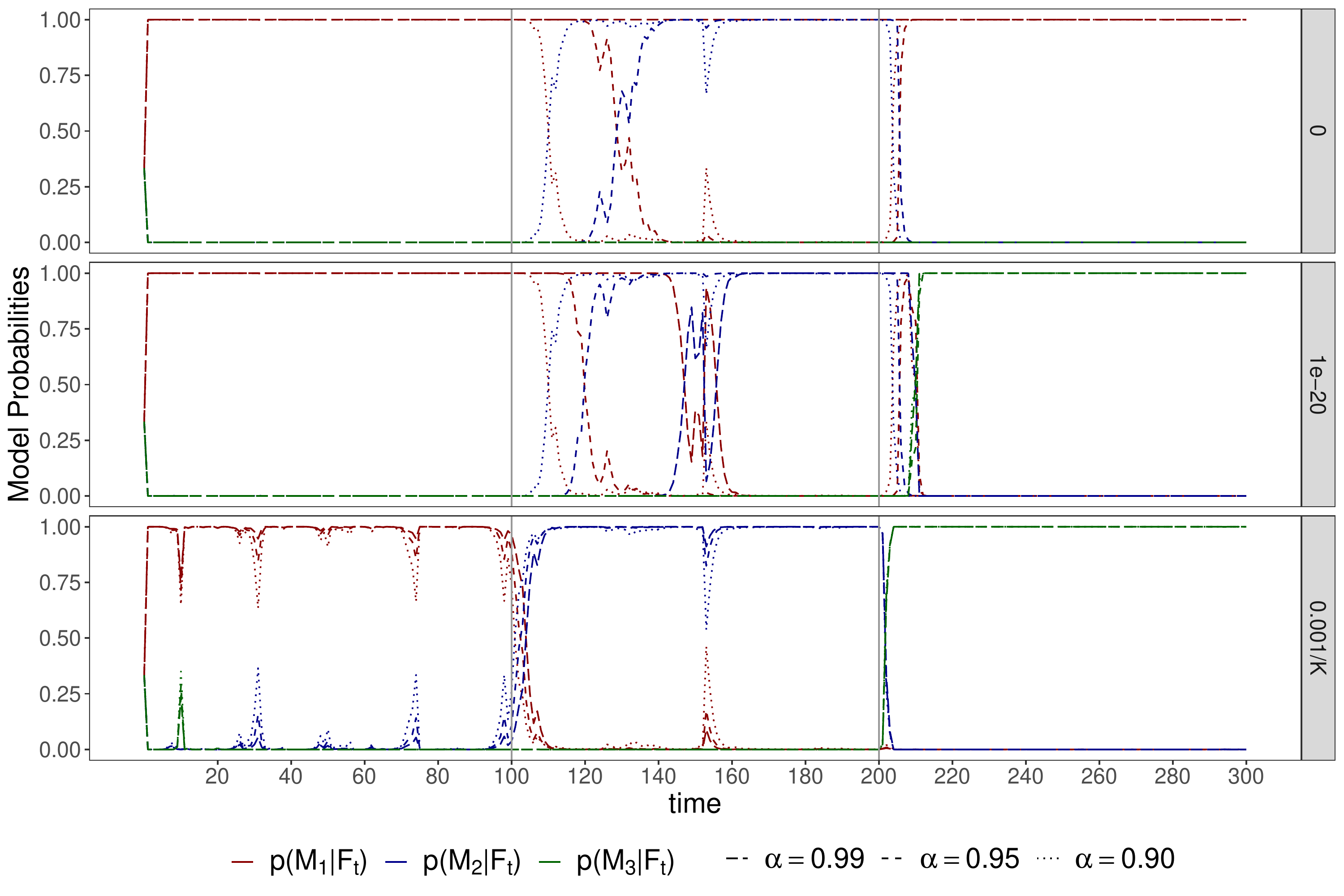}
\end{center}
\caption{Probabilities determined through Eq.~\eqref{eq:KK_16} for 
forgetting factor values, $\alpha \in \{0.99, 0.95,0.9\}$, and values of
~$c \in \{0,10^{-20},0.001/K\}$.}\label{fig:Weights1}
\end{figure}

Figure~\ref{fig:Weights1} depicts the evolution of the weights 
of the three models using red, blue and green colour respectively.
The three subfigures correspond to 
$c = 0,10^{-20},10^{-3}/3$, as recommended
by~\cite{CataniaNonejad_2018,KoopKorobilis_2012} and \cite{RafteryDMA2010},
respectively.
Within each subfigure a different type of dashed line is used to distinguish
between the three values of the forgetting factor, $\alpha \in \{0.99, 0.95,0.9\}$.
Note that $\alpha=0.9$ is much lower than any recommendation in the literature,
and is only included to explore the extent to which this forgetting factor enables
adaptation.
Finally, the two grey vertical lines depict the timing of the change points at
$t=100, 200$.

The top subfigure corresponds to~$c=0$. For this setting using $\alpha=0.99$
causes the probability of $M_1$ to be effectively equal to one throughout
the simulation despite the two change points. This occurs because the weights
of models $M_2$ and $M_3$ are equal to machine zero by time-step $t=100$.
A lower value of $\alpha$ allows the weights to adapt correctly to the first
change point, but the response is slow even when~$\alpha=0.9$.
However, by time-step $t=200$ the weight assigned to $M_3$ is equal
to zero for both $\alpha = 0.95,0.9$.
Therefore when $c=0$ not even very aggressive forgetting, $\alpha=0.9$, is
sufficient for Eq.~\eqref{eq:KK_16} to correctly identify the optimal model
after the second change point. Instead $M_1$ is assigned a weight of one.
Introducing a small constant, $c=10^{-20}$, (middle subfigure) enables the
weights to adapt correctly in response to both change points for all values
of $\alpha$. 
We see however that for~$\alpha=0.99$ approximately 50 time-steps are required
after the first change point for the weight assigned to $M_1$ to become
noticeably lower than one. 
Although the response to the second change point is faster, notice that for
$\alpha=0.95$ there is a short period immediately after the change point during
which the weight of $M_1$ (rather than $M_3$) increases abruptly.
The final subfigure corresponds to the case~$c=10^{-3}/3$. A larger constant
enables the weights to adapt much faster in response to both change points.
Notice for instance that for $c=10^{-3}/3$ and $\alpha=0.99$ model probabilities
adjust faster after the first change point, compared to the case when~$\alpha=0.9$ but
$c=10^{-20}$.
However, a larger value of $c$ also induces much higher variability during
periods in which the data generating process is static.

The above example demonstrates that the choice of~$c$ in Eq.~\eqref{eq:KK_16}
is at least as important as that of~$\alpha$.
Therefore any approach focused only on tuning $\alpha$ is not sufficient to
fully control the speed with which model probabilities in DMA adapt.
We propose to use the ConfHedge algorithm 
from the field of machine learning known as {\em prediction with expert
advice}~\citep{VyuginTrunov_2019}. As we discuss next ConfHedge has
very appealing theoretical properties and is parameter-free.

Prediction with expert advice studies the following online learning problem~\citep{CBianchiL2006}. At
time-step, $t+1$, each of the~$K$ forecasting models, called {\em experts}\/,
provides a forecast,~$\hat{y}^{(k)}_{t+1}$. An {\em aggregating algorithm}
predicts $\hat{y}_{t+1}$ through a convex combination of the experts'
forecasts. 
After observing $y_{t+1}$,
the weight of every expert is updated based on a measure of forecast error,
called {\em loss}\/. 
In a static environment, the goal is to design weight updates that guarantee
that the loss of the aggregating algorithm is never much larger than the cumulative loss
of the best expert, or the best convex combination of the losses of the
experts. 
In a dynamic environment the expert (or convex combination of experts) that
achieve the lowest loss may differ across segments of the time series, and
comparing against the best expert over the entire time series length 
can result in algorithms with poor forecast performance.
To address this issue consider a partition of the time series into at most $L+1$ segments,
$1<t_{(1)}<t_{(2)} < \cdots <t_{(L)}<T$, and allow the best expert (or convex
combination of experts) to differ across elements of this partition. The
best partition into at most $L+1$ segments is the one for which the optimal
sequence of experts (or convex combination of experts) achieves the lowest
cumulative loss.
The learning problem in this setting is considerably harder.
The ideal aggregating algorithm must achieve a loss that is as close as
possible to that of the sequence of experts (or convex combinations of experts)
that form the best partition of the time series into at most $L+1$ segments. 
Note that neither the maximum number of change points, $L$, nor the length of
each segment are known.

A number of algorithms have been proposed that achieve optimal upper
bounds for this problem, but these typically assume that the loss function is
uniformly bounded~\citep{HerbsterW1998}.
This assumption is not satisfied in our case since each DLM expert in ADMA includes a Gaussian error term.
ConfHedge is the first (and to the best of our
knowledge the only) method that upper bounds the loss of the aggregating
algorithm against an arbitrary sequence of experts (or convex combinations of
experts), when the loss function is unbounded~\citep{VyuginTrunov_2019}.

We next briefly describe the ConfHedge algorithm for our problem.
To distinguish from the DMA probabilities we denote as $w_{k,t+1}$ the weight
assigned by ConfHedge to model $M_k$ at the stage of predicting the response at
time $t+1$.
Initially the weights of all models are equal, $w_{1,t} = 1/K$.
At time-step $t$ the ConfHedge prediction $\hat{y}^{\text{CH}}_{t+1}$ is
obtained through,
\begin{align}
\hat{y}^{\text{CH}}_{t+1} = \sum_{k=1}^K \frac{w_{k,t+1}}{\sum_{m=1}^K
w_{m,t+1} } \hat{y}^{(k)}_{t+1}.
\end{align}
After observing $y_{t+1}$ the loss of every expert is estimated as the squared
forecast error,
\begin{align}\label{eq:lossL}
l_{k,t+1} = \frac{1}{2} \left(y_{t+1} - \hat{y}^{(k)}_{t+1} \right)^2,
\end{align}
and the loss of the aggregating algorithm is defined as,
\begin{align}\label{eq:lossH}
h_{t+1} = \sum_{k=1}^K w_{k,t+1} l_{k,t+1} = w_{t+1}^\top l_{t+1},
\end{align}
where $l_{t+1} = (l_{1,t+1},\ldots,l_{K,t+1})$.
The weights for time-step $t+2$ are updated according to,\footnote{We present
the update for the Fixed Share mixing scheme in~\cite{VyuginTrunov_2019}.}
\begin{align}\label{eq:VyuginW}
w^{\mu}_{k,t+1} & = \frac{ w_{k,t+1} e^{-\eta_{t+1} l_{k,t+1}  } }{
\sum_{m=1}^K w_{m,t+1} e^{-\eta_{t+1} l_{m,t+1} } }\\
w_{k,t+2} & = \frac{1}{(t+2)K} + \frac{t+1}{t+2} w^{\mu}_{k,t+1}.
\end{align}
At the first time-step, $\eta_1=\infty$, which implies $w^{\mu}_{k,1} = 1$ if
model $M_k$ achieves the minimum loss at time $t=1$, and zero
otherwise~\citep{Rooij2014}. In subsequent time-steps the step-size parameter
$\eta$ is updated according to,
\begin{align} \label{eq:VyuginPars}
m_{t+1} & = -\frac{1}{\eta_{t+1}} \log\left\{ \sum_{k=1}^K w_{k,t+1}
e^{-\eta_{t+1} l_{k,t+1} } \right\}, \\
\Delta_{t+1} & = \Delta_{t} + h_{t+1} - m_{t+1}, \quad \Delta_0=0,\\
\eta_{t+2} & = \max\{1,\log(K)\}/\Delta_{t+1}.
\end{align}
ConfHedge involves no user-defined parameters.
Eq.~\eqref{eq:VyuginW} shows that the algorithm explicitly uses a weighted
combination of two terms. The first term distributes a weight of $1/(t+2)$
equally among the~$K$ experts. Note that in the original DMA formulation, the
constant~$c$ in Eq.~\eqref{eq:KK_16} plays a similar role although in that case
it is not possible to control the proportion of the overall weight that is equally allocated
among all models.
The second term assigns a progressively larger proportion of the total weight
to experts that achieve lower loss.

Denote as $q_1,q_2,\ldots,q_T$ a sequence of convex combinations of experts,
where $q_t \in \left\{x \in \R^K_{+} \,|\, \sum_{k=1}^K x_{k} = 1 \right\}$.
The goal of the aggregating algorithm is to 
minimise the {\em shifting regret},
\begin{equation}\label{eq:regret}
R_T = \sum_{t=1}^T h_t - \sum_{t=1}^T q_{t}^\top l_{t},
\end{equation}
where $q_t$ changes at most $L$~times, at unknown time-points, $1<t_{(1)}<t_{(2)} < \cdots <t_{(L)}<T$.
\cite{VyuginTrunov_2019} prove a number of results that upper bound the shifting regret
of ConfHedge for finite~$T$. These bounds depend on the maximum number of change points,
$L$, as well as on the range of actual losses incurred by the individual forecasters.
The description of these results is beyond the scope of this paper. The following simple proposition
establishes that when the loss function is the one-step-ahead squared forecast 
error an upper bound on~$R_T$ 
translates to an upper bound on the mean squared forecast error.

\begin{prop}

Let $\{y_t\}_{t=1}^T$ be the observed time series, and $\{q_t^\star\}_{t=1}^T$
the sequence of weights that achieves the lowest cumulative squared forecast
error out of all sequences that contain at most~$L$ change points.
Consider the ConfHedge algorithm using the squared forecast error as loss function.
Let $\{ \hat{y}^{\mathrm{CH}}_t \}_{t=1}^T$ denote the ConfHedge predictions, and
$C$ the upper bound on the shifting regret established by~\cite{VyuginTrunov_2019}.
Then,
\[
\frac{1}{2} \sum_{t=1}^T \left(y_t - \hat{y}^{\mathrm{CH}}_t \right)^2 \leqslant C +
\sum_{t=1}^T (q_{t}^\star)^\top l_t.
\]

\end{prop}

\begin{proof}

The upper bound on the shifting regret by~\cite{VyuginTrunov_2019}
applies for any sequence $\{q_t\}_{t=1}^T$ with at most $L$ change points.
\begin{align*}
C + \sum_{t=1}^T (q_{t}^\star)^\top l_{t} & \geqslant  \sum_{t=1}^T w_t^\top l_t \\
& = \frac{1}{2} \sum_{t=1}^T  \sum_{k=1}^K w_{k,t} \left( y_t - y^{(k)}_t \right)^2 \\
& \geqslant \frac{1}{2}  \sum_{t=1}^T \left( y_t - \sum_{k=1}^K w_{k,t} y^{(k)}_t \right)^2\\
& =\frac{1}{2}  \sum_{t=1}^T  \left( y_t - \hat{y}^{\text{CH}}_t \right)^2.
\end{align*}
The second inequality follows from Jensen's inequality.

\end{proof}

\section{Simulation Experiments}\label{sec:Exp}

In this section, we investigate the behaviour of the proposed AF-DLM
on simulated data originating from static, gradually drifting and
abruptly changing data generating processes. 
We set the dimensionality of the covariate vector to five
to enable the visualisation of the path of the estimated coefficients.
In all cases, the covariates are sampled from a Gaussian distribution 
$x_t \sim \N(0,I)$, $t=1,\ldots,1000$, while the noise term in the measurement
equation has unit variance, $\sigma^2_t=1$ in Eq.~(\ref{eq:meas}).

\subsubsection*{Static Environment}

A static environment is characterised by a constant coefficient vector,
$\theta_t = \theta_0$.  We set $\theta_0=(-2,-1,1,2,3)^\top$ and obtain 100
time series of $y_t$ by randomly sampling from the measurement equation,
Eq.~(\ref{eq:meas}).
In Figure~\ref{fig1a} solid lines correspond to the actual values of each coefficient,
while dotted lines and shaded regions of the same colour represent the
median and the interquartile range of the estimated coefficient,
respectively.
The figure shows that $\est_t$ converges rapidly
to $\theta_0$, and the estimates exhibit very little variability across the 100~simulations.

\begin{figure}[!h]
\begin{center}
\subfigure[Coefficients]{\includegraphics[scale=0.24]{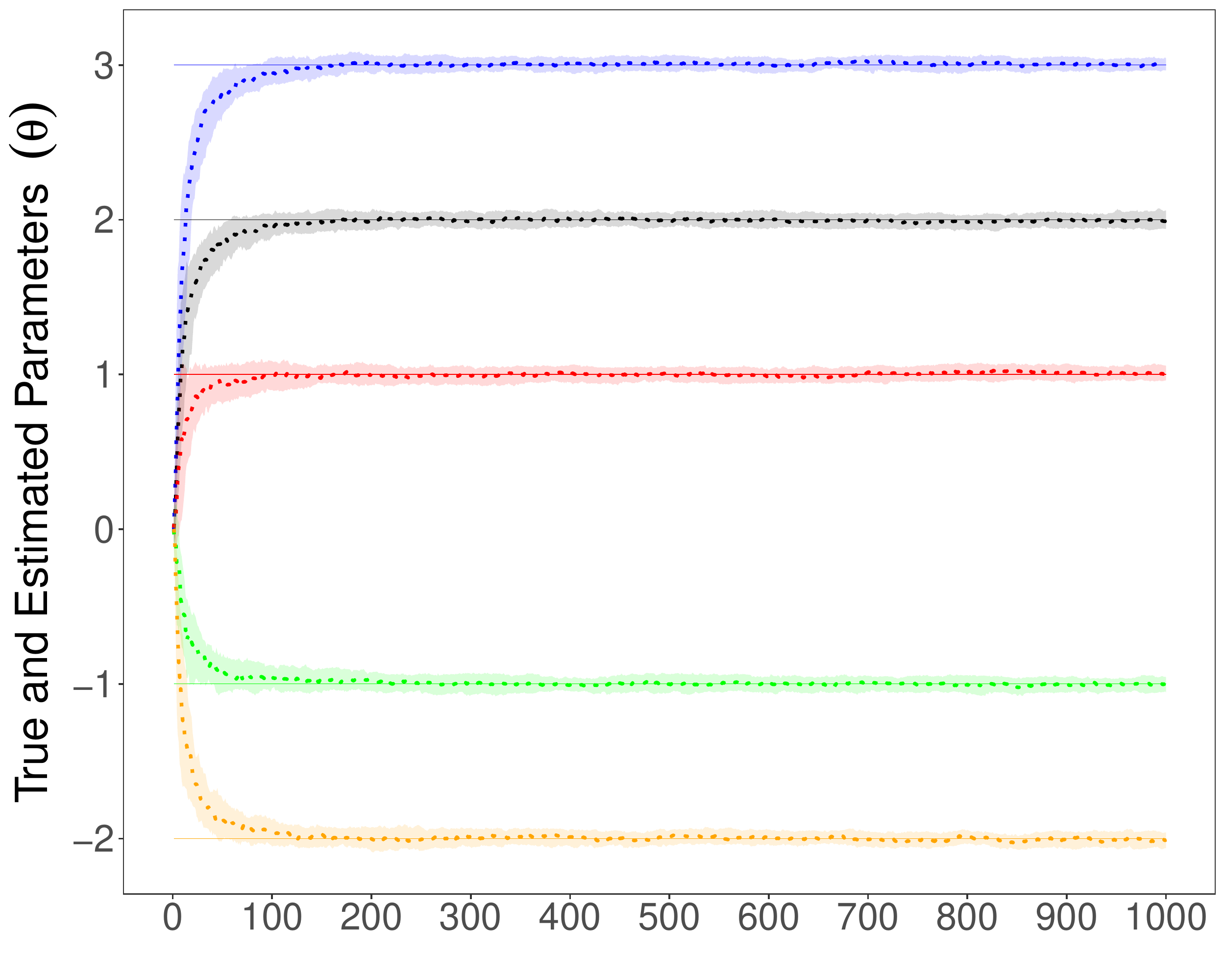}\label{fig1a}} \hspace{0.2cm}
\subfigure[Forgetting factor]{\includegraphics[scale=0.24]{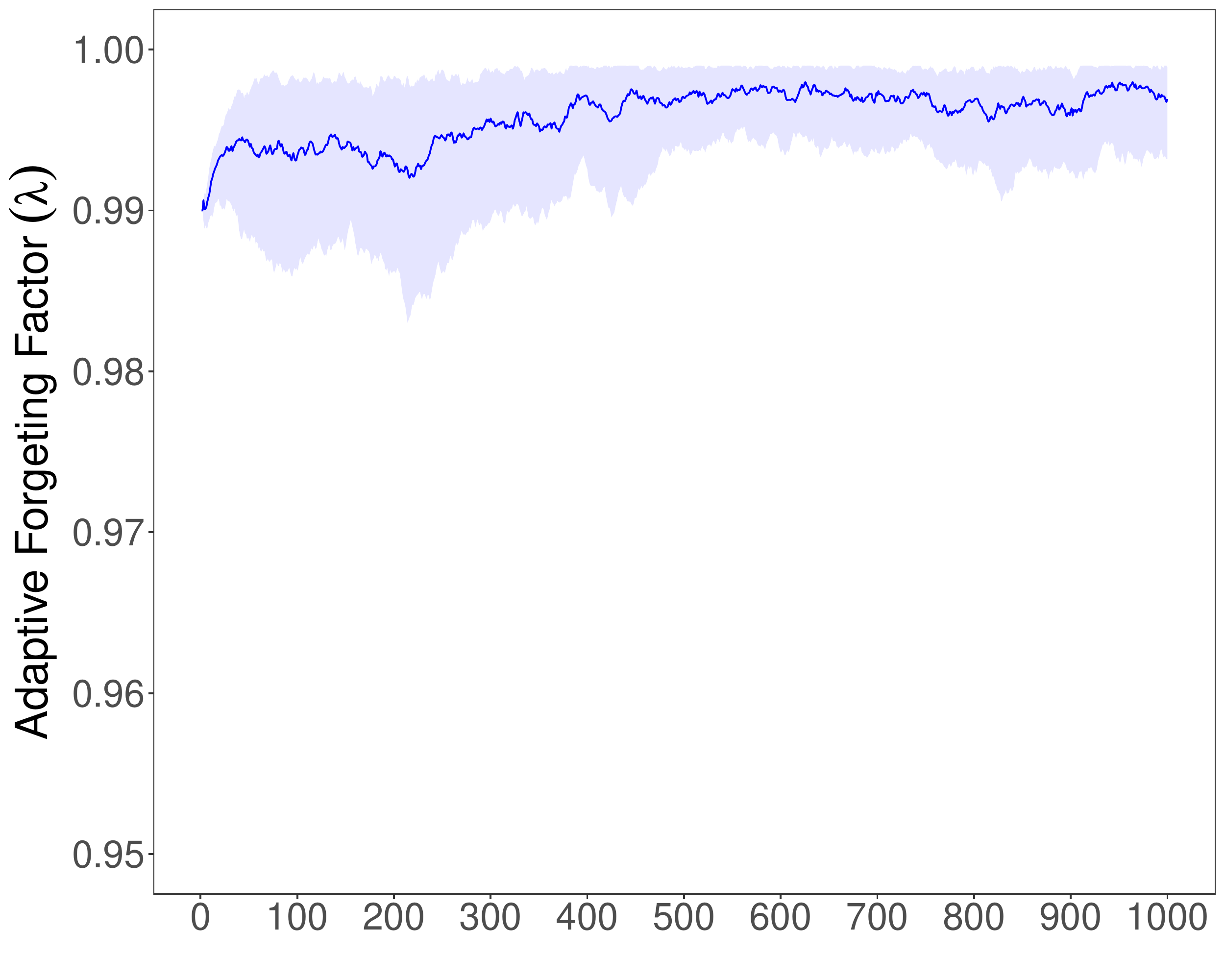}\label{fig1b}}
\end{center}
\caption{Evolution of estimated coefficients, and the
forgetting factor, $\lambda$, for a static data generating process.}
\label{fig:static}
\end{figure}

Figure~\ref{fig1b} illustrates the evolution of the forgetting factor,
$\lambda_t$, through the ADAM stochastic gradient descent algorithm. 
The value of $\lambda_t$ is high throughout the simulation, and as $\est_t$
converges to $\theta_0$ the median value of~$\lambda_t$ increases and the
variability across simulations decreases.
As the forgetting factor tends to unity, past and present examples become equally
weighted and consequently parameter estimates become more accurate and less
variable.

\subsubsection*{Abrupt Change}

Next, we consider dynamic environments in which the coefficient vector
changes at distinct time points, and remains constant in-between 
consecutive change points. 
We consider again time series of length 1000, and introduce change points at $t=100,400,700$.
As in the static environment, we simulate data from a single 
time series of $\theta_t$ and create
100 time series of $y_t$ by different realisations of the noise term in the
measurement equation, Eq~(\ref{eq:meas}).
The time series of the coefficient vector $\theta_t$ is specified by $\theta_0
= (3,2,1,-1,-2)^\top$ and, 
\[
\theta_t = \left\{ \begin{array}{ll} 0.5 \,
	\theta_{t-1} & \textrm{if } t=100, \\
	1.4 \, \theta_{t-1} & \textrm{if } t= 400, \\ 
	0.7 \, \theta_{t-1}, & \textrm{if } t = 700, \\
	\theta_{t-1}, & \textrm{otherwise}.  \end{array} \right.  
\]
The specific trajectory for~$\theta_t$ is selected because it allows
a clear visualisation of the evolution of $\est_t$ at each time-step.
Solid lines in Figure~\ref{fig:ab1} depict the evolution of $\theta_{t}$
while dotted lines and shaded areas of the same colour
correspond to the median estimated parameter and the associated interquartile
range, respectively.

At the first change point, $t=100$, the magnitude of the change in
every element of $\theta_t$ is the largest. Figure~\ref{fig:ab2} shows that
$\lambda_t$ decreases very rapidly in response to this,
and by the time step $t=200$ it assumes the smallest values observed during these simulations.
The minimum value of $\lambda_t$ observed is lower than 0.9 which
implies a very aggressive forgetting of past information, or equivalently
a very small effective window size.
As $\est_t$ approaches $\theta_t$, the forgetting factor steadily increases and
approaches its maximum value when the two almost coincide. This occurs right
before the second change point at $t=400$. The change in $\theta_t$
at $t=400$ is much smaller than that at the first change
point, and this is reflected in the evolution of the forgetting factor.
As Figure~\ref{fig:ab2} shows $\lambda_t$
decreases rapidly following the second change point but 
the lowest median value, which is close to 0.95, is much
higher than the corresponding minimum following the first change point. Subsequently,
the forgetting factor increases steadily as the estimated coefficients converge to the true values.
The third change point at $t=700$ reverts $\theta_t$ to its value prior to the second change point.
As Figure~\ref{fig:ab2} shows the effect of this change point on $\lambda_t$ 
is very similar to the pattern observed after the second change point. 

\begin{figure}[!h]
\begin{center}
\subfigure[Coefficients]{\includegraphics[scale=0.24]{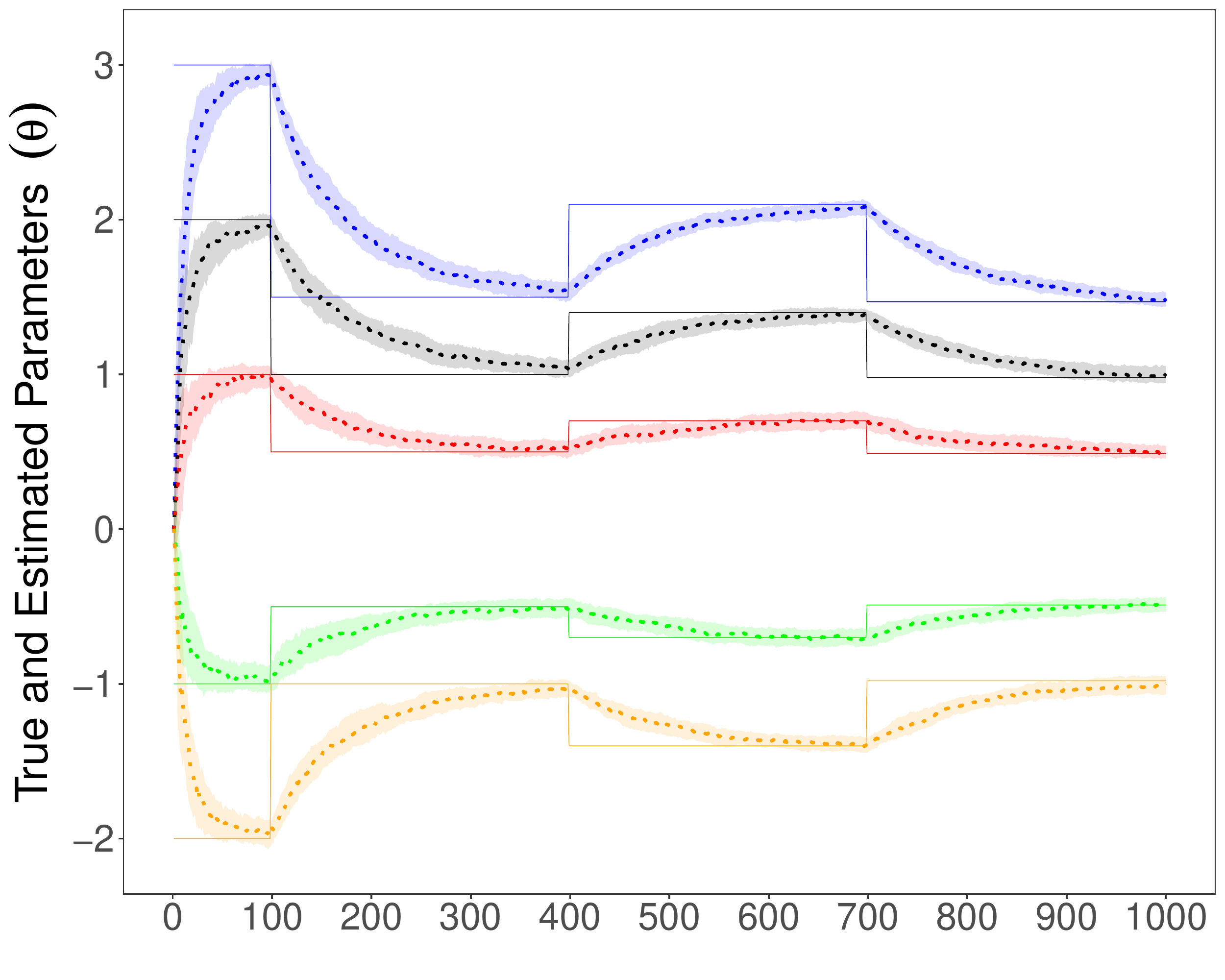}\label{fig:ab1}} \hspace{0.2cm}
\subfigure[Forgetting factor]{\includegraphics[scale=0.24]{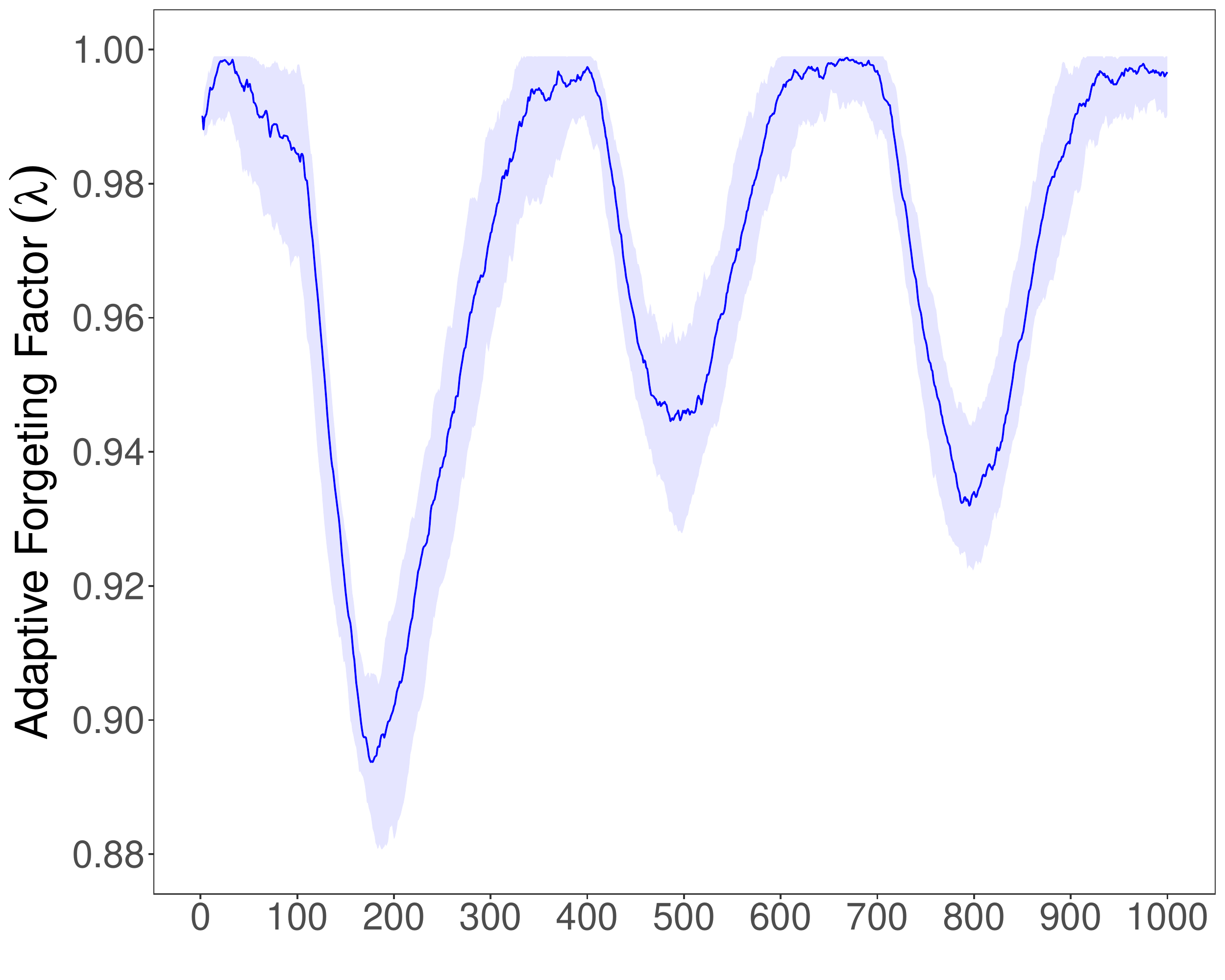}\label{fig:ab2}}
\end{center}
\caption{Evolution of estimated coefficients, and the
forgetting factor, $\lambda$, for a data generating process with abrupt changes.}
\label{fig:abrupt}
\end{figure}

Overall, the results depicted in Figure~\ref{fig:abrupt} indicate that the
proposed adaptive forgetting algorithm is capable of tuning $\lambda_t$
effectively in the presence of change points. Following each change point
AF-DLM induces a sharp decline in $\lambda_t$, which enables the estimated
coefficients to adjust rapidly. As $\est_t$ converges
to the true coefficients, which are static in-between consecutive change points,
$\lambda_t$ increases and, if the interval between consecutive change points
is sufficiently long, it approaches unity.

\subsubsection*{Gradual Drift}

Finally, we consider time series in which the coefficient vector changes
gradually over time.  For this purpose, we simulate from the state-space model
assumed by the DMA algorithm, namely Eqs.~(\ref{eq:state}) to~(\ref{eq:Wt}),
with $\theta_0 = 0$.
Our objective in this case is to evaluate whether the proposed adaptive
forgetting method can identify the true value of $\lambda$.
Note that in this case we are not able to simulate different realisations
of~$y_t$ for a single time series of~$\theta_t$, since 
the state transition covariance matrix,~$W_t$,
depends on the previous realisations of the forecast error.
\begin{figure}[!h]
\begin{center}
\includegraphics[scale=0.24]{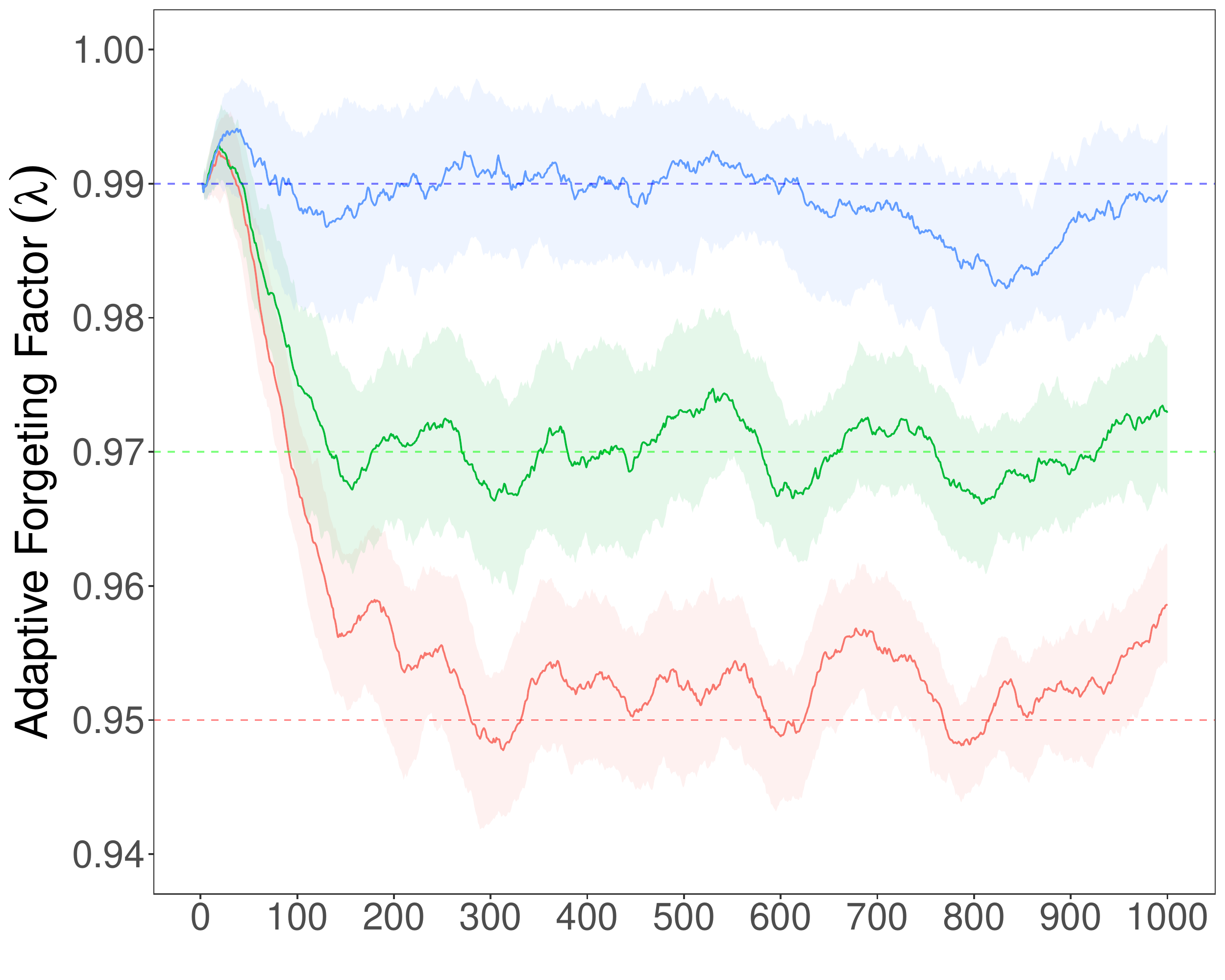}
\end{center}
\caption{Evolution of $\lambda_t$ 
	for time-series data sampled from the state-space model assumed by DMA
	for three different values of $\lambda$.}
\label{fig:driftL}
\end{figure}
We consider three values of the forgetting factor,~$\lambda \in \{0.99, 0.97,
0.95\}$, and for each value simulate 100 time series.
In Figure~\ref{fig:driftL} the dashed horizontal lines correspond to the true
values of $\lambda$, while the solid lines and the shaded regions of the same
colour depict the median and interquartile range of $\lambda_t$, respectively.
As the figure shows the proposed method rapidly adjusts $\lambda_t$ towards
the true value.
After the initial adjustment period $\lambda_t$ fluctuates around the true
value. This fluctuation is more variable for smaller values of $\lambda$,
which is consistent with the model since a smaller forgetting factor implies a higher variability in the trajectory of~$\theta_t$.
Note that beyond the initial adjustment period, the median value of $\lambda_t$ is never more than
0.01 away from the true value. Furthermore, the interquartile range
of $\lambda_t$ contains the true value in the vast majority of time-steps (the
only exceptions occur for $\lambda=0.95$ and their duration is short).

Overall, the results on simulated time series illustrate that AF-DLM can effectively
tune the value of the forgetting factor under different types of variation in
the data generating process. 
In response to abrupt changes AF-DLM decreases $\lambda$ sharply thereby enabling
the estimated coefficients to adjust rapidly. When the coefficients
are static
the adaptive forgetting factor 
tends to unity hence improving the accuracy and stability of the
estimation process. 
In cases where the coefficients change gradually the forgetting factor fluctuates
around a value that reflects the speed of drift. 
This concludes our empirical evaluation of the AF-DLM and in the next section
we focus on the comparative evaluation of ADMA on the 
problem of forecasting UK regional house prices.

\section{Forecasting UK Regional House Prices} \label{sec:House}

For our empirical application, we employ quarterly seasonally adjusted regional
house price indices for the period 1982:Q1 to 2017:Q4. The data is provided by
Nationwide, the largest building society in the world and one of the largest
mortgage providers in the UK. Following Nationwide's classification, we
consider 13 regional housing markets: the North, Yorkshire and Humberside,
North West, East Midlands, West Midlands, East Anglia, Outer South East, Outer
Metropolitan, Greater London, South West, Wales, Scotland and Northern Ireland.
To transform nominal into real prices, we divide by the consumer price index
(all items), obtained from the OECD Database of Main Economic Indicators, and
then compute the annualised log transformation of real property price inflation
as,
\begin{equation}\label{eq:data}
y_{r,t} = 400 \times \ln \left( \frac{P_{r,t}}{P_{r,t-1}} \right),
\hspace{0.4cm} r=1,\dotsc , 13,
\end{equation}
where $P_{r,t}$ stands for the level of the real house price index of market
$r$ at time $t$. 

For each region in our sample, we consider eleven economic variables as
potential predictors of future house price movements: four regional-level and
seven national-level predictors. The variables measured at the regional level
include the price-to-income ratio (which proxies for affordability), income
growth, the unemployment rate, and the growth in labour force. National-level
predictors consist of the real mortgage rate, the spread between yields on
long-term and short-term government securities, growth in industrial
production, the number of housing starts, growth in real consumption, the
Credit Conditions Index (CCI) proposed by~\cite{Corugedo_2006}, and a new
measure of House Price Uncertainty (HPU) which we construct using the news
based methodology of~\citet{BakerScott_2016}. For a description of the
variables, the data sources and the transformations undertaken we refer
the reader to Appendix~\ref{appendix:a}. 

From the above predictors, the first nine have been used
by~\cite{BorkMoller_2015} to forecast house price movements in US metropolitan
states. The last two have not been employed in a forecasting context before but
may well have predictive content for future house price inflation. With regard
to CCI, credit supply conditions in the UK economy, especially in the mortgage
market, have changed dramatically since the 1970s. As argued by several
authors, such changes were at the heart of the housing boom that preceded the
Great Recession. It therefore seems natural to investigate whether an index of
credit conditions may contain valuable information for forecasting.\footnote{A
deficiency of simple proxies for credit conditions, such as interest-rate
spreads and unsecured credit to income ratios, is that they fail to control for
the economic environment, and are thus subject to an endogeneity problem. The
methodology of \cite{Corugedo_2006} mitigates this problem by making use of a
large number of economic and demographic controls.} Similarly, changes in house
price uncertainty  impact on housing investment and real estate construction
decisions \citep{Cunningham2006,BanksBOS2015,OhY2019}, and thus may lead to
future house price movements. 

In addition to macro and financial variables, there is a substantial empirical
literature that documents the existence of strong spatial linkages between UK
regional markets in sample \citep[see, e.g.,][\textit{inter alia}]{Drake_1995,
Meen_1999, CookThomas_2003, Holly_2010,Antonakakis_CFG2018}. To accommodate
this, we incorporate in the set of potential predictors lagged property price
growth in contiguous regions. The number of neighbouring regions for each of
the 13 real estate markets under consideration lies in the range of one to
five.

\subsection{In-Sample Evidence of Structural Instability}

Before proceeding to the forecasting exercise, we examine whether there is
evidence of structural instability in the relationship between real house price
inflation and individual house price predictors in sample. To do so, we employ
two tests proposed by \citet{ChenH2012}. The first is a Hausman-type ($H$) test
that compares time-varying parameter estimates obtained by local linear
regression to constant estimates obtain by ordinary least squares. The second
is a Chow-type ($C$) test which compares the sum of squared residuals between
the constant parameter and local linear regression models. The null hypothesis
in both tests is that of time-invariant regression coefficients. 

The $H$ and $C$ tests have a number of attractive features. First, because they
impose minimal restrictions on the functional form of the time-varying
parameters, they are consistent with both smooth and abrupt structural change
and, from this perspective, correspond well to the AF-DLM of Section
\ref{sec:meth}. Second, they require no prior information regarding the timing
and the number of breaks. Third, they are asymptotically pivotal and, fourth,
they do not involve trimming of the boundary region near the end points of the
sample period.

\begin{table}[h]
\caption{Stability Test Results}
\label{stability1}
\hspace{-2.0cm}\begin{tabular}{lcccccccccccccc}
  \hline\hline
   & \multicolumn{2}{c}{\scriptsize EA} & \multicolumn{2}{c}{\scriptsize EM} & \multicolumn{2}{c}{\scriptsize GL} & \multicolumn{2}{c}{\scriptsize NI} & \multicolumn{2}{c}{\scriptsize NT} & \multicolumn{2}{c}{\scriptsize NW} & \multicolumn{2}{c}{\scriptsize OM}   \\
   & \scriptsize $H$ & \scriptsize $C$  & \scriptsize $H$ & \scriptsize $C$ & \scriptsize $H$ & \scriptsize $C$  & \scriptsize $H$ & \scriptsize $C$ & \scriptsize $H$ & \scriptsize $C$ & \scriptsize $H$ & \scriptsize $C$ & \scriptsize $H$ & \scriptsize $C$\\ [1 ex] \hline 
    \multicolumn{15}{c}{ \small Univariate Predictor Regressions} \\

\scriptsize RATIO & \scriptsize 0.0011 & \scriptsize 0.0015   &\scriptsize 0.0018 &  \scriptsize  0.0019 & \scriptsize 0.0003 & \scriptsize 0.0000 & \scriptsize 0.0000 & \scriptsize 0.0000 & \scriptsize 0.0169 & \scriptsize 0.0013 & \scriptsize 0.0808 & \scriptsize 0.0322 & \scriptsize 0.0001 & \scriptsize 0.0002  \\

\scriptsize GROWTH & \scriptsize 0.0006 &  \scriptsize 0.0008   & \scriptsize 0.0008  & \scriptsize 0.0000   & \scriptsize 0.0000  & \scriptsize 0.0000  & \scriptsize 0.0000 & \scriptsize 0.0000 & \scriptsize 0.0017 & \scriptsize 0.0011 & \scriptsize 0.0002 & \scriptsize 0.0000 & \scriptsize 0.0000  & \scriptsize 0.0000 \\ 
 
\scriptsize UR & \scriptsize 0.0000 & \scriptsize 0.0000 & \scriptsize 0.0003  & \scriptsize 0.0000  & \scriptsize 0.0000  & \scriptsize 0.0000 & \scriptsize 0.0018 & \scriptsize 0.0024 & \scriptsize 0.0000 & \scriptsize 0.0000 & \scriptsize 0.0000 & \scriptsize 0.0000 & \scriptsize 0.0000 & \scriptsize 0.0000  \\ 

\scriptsize LF & \scriptsize 0.0000 & \scriptsize 0.0000   & \scriptsize 0.0000  & \scriptsize 0.0000  & \scriptsize 0.0000 & \scriptsize 0.0000 & \scriptsize 0.0202 & \scriptsize 0.0065 & \scriptsize 0.0039 & \scriptsize 0.0020 & \scriptsize 0.0001 & \scriptsize 0.0000 & \scriptsize 0.0001 & \scriptsize 0.0000 \\ 

\scriptsize HS & \scriptsize 0.0000 & \scriptsize 0.0000  & \scriptsize 0.0001  & \scriptsize 0.0001  & \scriptsize 0.0000 & \scriptsize 0.0000  & \scriptsize 0.0014 & \scriptsize 0.0127 & \scriptsize 0.0052 & \scriptsize 0.0004 & \scriptsize 0.0002 & \scriptsize 0.0000 & \scriptsize 0.0000 & \scriptsize 0.0000 \\ 

\scriptsize CONS & \scriptsize 0.0001 & \scriptsize 0.0001  & \scriptsize 0.0001  & \scriptsize 0.0000 & \scriptsize 0.0000  & \scriptsize 0.0000  & \scriptsize 0.0002 & \scriptsize 0.0009 & \scriptsize 0.0034 & \scriptsize 0.0061 & \scriptsize 0.0001 & \scriptsize 0.0003 & \scriptsize 0.0000 & \scriptsize 0.0000 \\ 
 
\scriptsize INDUS & \scriptsize 0.0003 & \scriptsize 0.0000 & \scriptsize 0.0000 & \scriptsize 0.0000 & \scriptsize 0.0000  & \scriptsize 0.0000 & \scriptsize 0.0018 & \scriptsize 0.0023 & \scriptsize 0.0002 & \scriptsize 0.0003 & \scriptsize 0.0000 & \scriptsize 0.0000 & \scriptsize 0.0000 & \scriptsize 0.0000 \\ 
 
\scriptsize RABMR & \scriptsize 0.0004 & \scriptsize 0.0000   & \scriptsize 0.0006 & \scriptsize 0.0001 & \scriptsize 0.0000 & \scriptsize 0.0000 & \scriptsize 0.0222 & \scriptsize 0.0033  & \scriptsize 0.0015 & \scriptsize 0.0000 & \scriptsize 0.0000 & \scriptsize 0.0000 & \scriptsize 0.0000 & \scriptsize 0.0000 \\ 

\scriptsize SPREAD & \scriptsize 0.0000 & \scriptsize 0.0000   & \scriptsize 0.0000  & \scriptsize 0.0000  & \scriptsize 0.0000 & \scriptsize 0.0000  & \scriptsize 0.0752 & \scriptsize 0.0859 & \scriptsize 0.0000 & \scriptsize 0.0001 & \scriptsize 0.0000 & \scriptsize 0.0000 & \scriptsize 0.0000 & \scriptsize 0.0000 \\ 
 
\scriptsize CCI & \scriptsize 0.0000 & \scriptsize 0.0000   & \scriptsize 0.0027 & \scriptsize 0.0000  & \scriptsize 0.0000 & \scriptsize 0.0000 & \scriptsize 0.0000 & \scriptsize 0.0008 & \scriptsize 0.0197 & \scriptsize 0.0034 & \scriptsize 0.0028 & \scriptsize 0.0004 & \scriptsize 0.0000 & \scriptsize 0.0000 \\ 

\scriptsize HPU & \scriptsize 0.0001 & \scriptsize 0.0199  & \scriptsize 0.0000  & \scriptsize 0.0000  & \scriptsize 0.0000 & \scriptsize 0.0000 & \scriptsize 0.0000 & \scriptsize 0.0000 & \scriptsize 0.0012 & \scriptsize 0.0037 & \scriptsize 0.0000 & \scriptsize 0.0000 & \scriptsize 0.0000 & \scriptsize 0.0003 \\  \hline
\end{tabular}
\begin{tablenotes}
\footnotesize
\item Notes: The table reports wild-bootstrap $p$-values, based on $B = 9999$ iterations, of the Hausman- ($H$) and Chow-type ($C$) structural stability tests of \citet{ChenH2012}.  
\end{tablenotes}
\end{table}

 \begin{table}
\caption{Stability Test Results (Cont.)}
\label{stability2}
\centering
\begin{tabular}{lcccccccccccc}
  \hline\hline
   & \multicolumn{2}{c}{\scriptsize OSE} & \multicolumn{2}{c}{\scriptsize SC} & \multicolumn{2}{c}{\scriptsize SW} & \multicolumn{2}{c}{\scriptsize WM} & \multicolumn{2}{c}{\scriptsize WW} & \multicolumn{2}{c}{\scriptsize YH}  \\
   & \scriptsize $H$ & \scriptsize $C$  & \scriptsize $H$ & \scriptsize $C$ & \scriptsize $H$ & \scriptsize $C$  & \scriptsize $H$ & \scriptsize $C$ & \scriptsize $H$ & \scriptsize $C$ & \scriptsize $H$ & \scriptsize $C$  \\ [1 ex] \hline 
    \multicolumn{13}{c}{ \small Univariate Predictor Regressions} \\

\scriptsize RATIO & \scriptsize 0.0000 & \scriptsize 0.0000   &\scriptsize 0.0037 &  \scriptsize  0.0058 & \scriptsize 0.0180 & \scriptsize 0.0018 & \scriptsize 0.0181 & \scriptsize 0.0112 & \scriptsize 0.0000 & \scriptsize 0.0000 & \scriptsize 0.0043 & \scriptsize 0.0017   \\

\scriptsize GROWTH & \scriptsize 0.0000 &  \scriptsize 0.0000   & \scriptsize 0.0002  & \scriptsize 0.0000   & \scriptsize 0.0002  & \scriptsize 0.0002  & \scriptsize 0.0003 & \scriptsize 0.0001 & \scriptsize 0.0023 & \scriptsize 0.0013 & \scriptsize 0.0030 & \scriptsize 0.0078  \\ 
 
\scriptsize UR & \scriptsize 0.0000 & \scriptsize 0.0000 & \scriptsize 0.0015  & \scriptsize 0.0000  & \scriptsize 0.0000  & \scriptsize 0.0000 & \scriptsize 0.0000 & \scriptsize 0.0000 & \scriptsize 0.0001 & \scriptsize 0.0000 & \scriptsize 0.0000 & \scriptsize 0.0000  \\ 

\scriptsize LF & \scriptsize 0.0000 & \scriptsize 0.0000   & \scriptsize 0.0012  & \scriptsize 0.0003  & \scriptsize 0.0000 & \scriptsize 0.0000 & \scriptsize 0.0006 & \scriptsize 0.0002 & \scriptsize 0.0000 & \scriptsize 0.0000 & \scriptsize 0.0059 & \scriptsize 0.0019  \\ 

\scriptsize HS & \scriptsize 0.0000 & \scriptsize 0.0000  & \scriptsize 0.0192  & \scriptsize 0.0357  & \scriptsize 0.0000 & \scriptsize 0.0000  & \scriptsize 0.0000 & \scriptsize 0.0000 & \scriptsize 0.0041 & \scriptsize 0.0013 & \scriptsize 0.0008 & \scriptsize 0.0009  \\ 

\scriptsize CONS & \scriptsize 0.0000 & \scriptsize 0.0000  & \scriptsize 0.0018  & \scriptsize 0.0048 & \scriptsize 0.0000  & \scriptsize 0.0000  & \scriptsize 0.0003 & \scriptsize 0.0000 & \scriptsize 0.0010 & \scriptsize 0.0008 & \scriptsize 0.0005 & \scriptsize 0.0016  \\ 
 
\scriptsize INDUS & \scriptsize 0.0000 & \scriptsize 0.0000 & \scriptsize 0.0000 & \scriptsize 0.0001 & \scriptsize 0.0000  & \scriptsize 0.0000 & \scriptsize 0.0000 & \scriptsize 0.0000 & \scriptsize 0.0242 & \scriptsize 0.0057 & \scriptsize 0.0002 & \scriptsize 0.0003  \\ 
 
\scriptsize RABMR & \scriptsize 0.0000 & \scriptsize 0.0000   & \scriptsize 0.0006 & \scriptsize 0.0000 & \scriptsize 0.0000 & \scriptsize 0.0000 & \scriptsize 0.0008 & \scriptsize 0.0000  & \scriptsize 0.0002 & \scriptsize 0.0000 & \scriptsize 0.0134 & \scriptsize 0.0049   \\ 

\scriptsize SPREAD & \scriptsize 0.0000 & \scriptsize 0.0000   & \scriptsize 0.0012  & \scriptsize 0.0004  & \scriptsize 0.0000 & \scriptsize 0.0000  & \scriptsize 0.0000 & \scriptsize 0.0000 & \scriptsize 0.0000 & \scriptsize 0.0000 & \scriptsize 0.0000 & \scriptsize 0.0000  \\ 
 
\scriptsize CCI & \scriptsize 0.0000 & \scriptsize 0.0000   & \scriptsize 0.0001 & \scriptsize 0.0003  & \scriptsize 0.0000 & \scriptsize 0.0000 & \scriptsize 0.0000 & \scriptsize 0.0000 & \scriptsize 0.0446 & \scriptsize 0.0008 & \scriptsize 0.0408 & \scriptsize 0.0066 \\ 

\scriptsize HPU & \scriptsize 0.0000 & \scriptsize 0.0013  & \scriptsize 0.0000  & \scriptsize 0.0010  & \scriptsize 0.0001 & \scriptsize 0.0008 & \scriptsize 0.0000 & \scriptsize 0.0014 & \scriptsize 0.0003 & \scriptsize 0.0026 & \scriptsize 0.0000 & \scriptsize 0.0003  \\  \hline
 \end{tabular}
\begin{tablenotes}
\footnotesize
\item Notes: The table reports wild-bootstrap $p$-values, based on $B = 9999$ iterations, of the Hausman- ($H$) and Chow-type ($C$) structural stability tests of \citet{ChenH2012}.
\end{tablenotes}
\end{table}

Tables \ref{stability1} and  \ref{stability2} report wild-bootstrap $p$-values of the $H$ and $C$ tests for each of the 11 house price predictors considered and for each of the 13 regions. We observe that, out of the 286 $p$-values, none exceeds 10 percent, four exceed five percent, and the vast majority lie below the one percent threshold. This strong evidence of structural instability motivates the use of dynamic econometric models for forecasting house price inflation. 

\subsection{Comparison of Forecast Accuracy}

We begin our out-of-sample analysis by comparing the forecast accuracy of a battery of econometric models relative to the AR(1) benchmark as well as the performance of ADMA relative to each of the remaining models in the pool. This set of models consists of the DMA formulation of \cite{DanglH2012} (abbreviated as eDMA due to the {\tt R} implementation of \citet{CataniaNonejad_2018}), two versions of the DMA of~\cite{KoopKorobilis_2012} (one with relatively slow forgetting, $\lambda=\alpha=0.99$, and another with fast forgetting $\lambda=\alpha=0.95$), a single DLM with $\lambda=0.99$ that includes all available predictors and, finally, Bayesian Model Averaging (BMA). Table \ref{Models} provides an overview of all models. 

\begin{table*}
\caption{An Overview of the Alternative Forecasting Strategies}
\label{Models}
	\begin{tabular}{p{1.5cm}p{15cm}}
	 \hline
	\textbf{ADMA} &  Adaptive Dynamic Model Averaging which uses adaptive forgetting and the aggregating algorithm of~\citet{VyuginTrunov_2019}  \\
	\textbf{eDMA} &  Dynamic Model Averaging which uses the grid of values   $(0.90,0.91,\dotsc,0.99)$ for $\lambda$ and $\alpha = 1$ \citep{DanglH2012} \\
	\textbf{DMA}$_{0.99}$ & Dynamic Model Averaging with $\lambda=\alpha=0.99$  \citep{RafteryDMA2010,KoopKorobilis_2012}\\
	\textbf{DMA}$_{0.95}$ & Dynamic Model Averaging with $\lambda=\alpha=0.95$ \citep{KoopKorobilis_2012} \\
	\textbf{BMA} & Bayesian Model Averaging \citep{Hoeting1999} \\
	\textbf{DLM$_{0.99}$} & A single time-varying parameter model with $\lambda=0.99$ that includes all the predictors \\
	\textbf{AR(1)} &  Recursive AR(1) model \\	\hline	 
	\end{tabular}
\end{table*}

Table~\ref{ForecastAccuracy} summarizes the forecasting performance of each model relative to the AR(1) benchmark over the out-of-sample evaluation period, which runs from 1995:Q1 to 2017:Q4. The second column of the table provides the realised Mean Squared Forecast Error (MSFE) of the AR(1) model for each region, while the remaining columns report the ratio of the MSFE of the competing models to that of the AR(1). Tables~\ref{ADMAvsDMA} and  \ref{DMAvsADMA} report MFSE ratios when ADMA is set as the alternative and the benchmark model, respectively. In all three tables, a~$\dagger$ indicates cases when the test of~\cite{ClarkWest_2007} rejects the null hypothesis of equal predictive accuracy in favour of the one-sided alternative that the competing model outperforms the benchmark at the 5\% significance level. 

\begin{table*}
\caption{Summary of Forecasting Performance}
\label{ForecastAccuracy}
\centering
\begin{tabular}{lccccccc}
  \hline\hline
  \tiny (1) & \tiny (2) & \tiny (3) & \tiny (4) & \tiny (5) & \tiny (6) & \tiny (7) & \tiny (8)  \\
\scriptsize Region & \scriptsize AR(1) & \scriptsize ADMA  & \scriptsize eDMA & \scriptsize DMA$_{0.99}$ & \scriptsize DMA$_{0.95}$ & \scriptsize BMA  & \scriptsize DLM$_{0.99}$ \\ [1 ex] \hline

\scriptsize East Anglia & \scriptsize 101.05 & \scriptsize \textbf{0.80} $\dagger$  &\scriptsize 1.04 &  \scriptsize  1.03 & \scriptsize 1.24 & \scriptsize 1.32  & \scriptsize 1.39\\

\scriptsize East Midlands & \scriptsize 70.74 &  \scriptsize 0.80 $\dagger$  & \scriptsize \textbf{0.79} $\dagger$ & \scriptsize 0.84 $\dagger$  & \scriptsize 0.83 $\dagger$ & \scriptsize \textbf{0.79} $\dagger$ & \scriptsize 1.44 \\

\scriptsize Greater London & \scriptsize 117.29 & \scriptsize 0.74 $\dagger$ & \scriptsize 0.79 $\dagger$ & \scriptsize 0.79 $\dagger$ & \scriptsize 0.77 $\dagger$ & \scriptsize \textbf{0.73}$\dagger$ & \scriptsize 0.77 $\dagger$  \\

\scriptsize Northern Ireland & \scriptsize 292.44 & \scriptsize \textbf{0.89} $\dagger$  & \scriptsize 0.97 $\dagger$ & \scriptsize 0.92 $\dagger$ & \scriptsize 1.00 & \scriptsize 1.02 & \scriptsize 0.94 $\dagger$\\

\scriptsize North & \scriptsize 167.39 & \scriptsize 0.72 $\dagger$  & \scriptsize 0.75 $\dagger$ & \scriptsize 0.78 $\dagger$ & \scriptsize 0.74 $\dagger$ & \scriptsize \textbf{0.69} $\dagger$ & \scriptsize 0.78 $\dagger$ \\

\scriptsize North West & \scriptsize 64.59 & \scriptsize 0.93 $\dagger$  & \scriptsize \textbf{0.92} $\dagger$ & \scriptsize 1.02 & \scriptsize 0.94 $\dagger$ & \scriptsize 0.94 $\dagger$ & \scriptsize 1.11 \\

\scriptsize Outer Metropolitan & \scriptsize 54.51 & \scriptsize \textbf{0.95}$\dagger$ & \scriptsize 1.04 & \scriptsize 1.14 & \scriptsize 0.96 $\dagger$ & \scriptsize \textbf{0.95}$\dagger$ & \scriptsize 1.16  \\

\scriptsize Outer South East & \scriptsize 61.84 & \scriptsize \textbf{0.92} $\dagger$  & \scriptsize 1.01 & \scriptsize 1.06 & \scriptsize 1.05 & \scriptsize 1.06 & \scriptsize 1.33\\

\scriptsize Scotland & \scriptsize 78.18 & \scriptsize 0.94 $\dagger$  & \scriptsize \textbf{0.90} $\dagger$ & \scriptsize 0.99 $\dagger$ & \scriptsize 1.03 & \scriptsize 0.96 $\dagger$ & \scriptsize 0.99 $\dagger$ \\

\scriptsize South West & \scriptsize 60.25 & \scriptsize 0.99 $\dagger$  & \scriptsize 0.97 & \scriptsize \textbf{0.98} $\dagger$ & \scriptsize 1.15 & \scriptsize 1.13 & \scriptsize 1.58 \\

\scriptsize West Midlands & \scriptsize 53.34 & \scriptsize \textbf{0.85} $\dagger$  & \scriptsize 0.97 $\dagger$ & \scriptsize 0.95 $\dagger$ & \scriptsize 0.94 $\dagger$ & \scriptsize 0.94 $\dagger$ & \scriptsize 1.44 \\

\scriptsize Wales & \scriptsize 148.29 & \scriptsize 0.81 $\dagger$  & \scriptsize \textbf{0.79} $\dagger$ & \scriptsize 0.90 $\dagger$ & \scriptsize 0.80 $\dagger$ & \scriptsize 0.80 $\dagger$ & \scriptsize 0.87 $\dagger$\\

\scriptsize Yorkshire \& Humber & \scriptsize 91.75 & \scriptsize \textbf{0.78}$\dagger$  & \scriptsize 0.79 $\dagger$ & \scriptsize 0.88 $\dagger$ & \scriptsize 0.80 $\dagger$ & \scriptsize 0.79 $\dagger$ & \scriptsize 0.82 $\dagger$\\ \hline

\end{tabular}
\begin{tablenotes}
\footnotesize
\item Notes: The second column of the table reports the realised MSFE of the AR(1) model. For the remaining models, the table reports the ratio of their realised MSFE to that of the AR(1). The forecasting model with the lowest MSFE is in bold. A $\dagger$ indicates rejection of the null hypothesis of the~\cite{ClarkWest_2007}  test at the 5\% significance level.
\end{tablenotes}
\end{table*}

\begin{table*}
\caption{Forecasting Performance of ADMA Relative to Alternative Forecasting Strategies}
\label{ADMAvsDMA}
\centering
\begin{tabular}{llllll}
  \hline\hline
  \tiny (1) & \tiny (2) & \tiny (3) & \tiny (4) & \tiny (5) & \tiny (6) \\
\scriptsize Region & \scriptsize eDMA & \scriptsize DMA$_{0.99}$ & \scriptsize DMA$_{0.95}$ & \scriptsize BMA & \scriptsize DLM$_{0.99}$ \\ \hline

\scriptsize East Anglia & \scriptsize \textbf{0.78} $\dagger$ & \scriptsize \textbf{0.79} $\dagger$ & \scriptsize \textbf{0.66} $\dagger$ & \scriptsize \textbf{0.61} $\dagger$ & \scriptsize \textbf{0.58} $\dagger$\\

\scriptsize East Midlands & \scriptsize 1.02 & \scriptsize \textbf{0.96} $\dagger$ & \scriptsize \textbf{0.97} $\dagger$ & \scriptsize 1.02 & \scriptsize \textbf{0.56} $\dagger$ \\

\scriptsize Greater London & \scriptsize \textbf{0.94} $\dagger$ & \scriptsize \textbf{0.93} $\dagger$ & \scriptsize \textbf{0.96} $\dagger$ & \scriptsize 1.01 & \scriptsize \textbf{0.97} $\dagger$\\

\scriptsize Northern Ireland & \scriptsize \textbf{0.92} & \scriptsize \textbf{0.95} & \scriptsize \textbf{0.87} $\dagger$ & \scriptsize \textbf{0.88} $\dagger$ & \scriptsize \textbf{0.94} $\dagger$ \\

\scriptsize North & \scriptsize \textbf{0.95} $\dagger$ & \scriptsize \textbf{0.91} $\dagger$ & \scriptsize \textbf{0.97} $\dagger$ & \scriptsize 1.03 & \scriptsize \textbf{0.92} $\dagger$ \\

\scriptsize North West & \scriptsize 1.01 & \scriptsize \textbf{0.91} $\dagger$  & \scriptsize \textbf{0.99} & \scriptsize \textbf{0.99} & \scriptsize \textbf{0.83} $\dagger$\\

\scriptsize Outer Metropolitan & \scriptsize \textbf{0.92} $\dagger$ & \scriptsize \textbf{0.84} $\dagger$ & \scriptsize \textbf{0.99} & \scriptsize 1.00 & \scriptsize \textbf{0.82} $\dagger$ \\

\scriptsize Outer South East & \scriptsize \textbf{0.90} $\dagger$ & \scriptsize \textbf{0.87} $\dagger$ & \scriptsize \textbf{0.89} $\dagger$ & \scriptsize \textbf{0.86} $\dagger$ & \scriptsize \textbf{0.69} $\dagger$ \\

\scriptsize Scotland & \scriptsize 1.04 & \scriptsize \textbf{0.95} $\dagger$ & \scriptsize \textbf{0.91}$\dagger$ & \scriptsize \textbf{0.98} & \scriptsize \textbf{0.94} $\dagger$ \\

\scriptsize South West & \scriptsize 1.02 & \scriptsize 1.01 & \scriptsize \textbf{0.86} $\dagger$ & \scriptsize \textbf{0.88} $\dagger$ & \scriptsize \textbf{0.63} $\dagger$  \\

\scriptsize West Midlands & \scriptsize \textbf{0.87} $\dagger$ & \scriptsize \textbf{0.89} $\dagger$ & \scriptsize \textbf{0.89} $\dagger$ & \scriptsize \textbf{0.89}  $\dagger$ & \scriptsize \textbf{0.58} $\dagger$ \\

\scriptsize Wales & \scriptsize 1.02 & \scriptsize \textbf{0.89} $\dagger$ & \scriptsize 1.00 & \scriptsize 1.01 & \scriptsize \textbf{0.93} $\dagger$\\

\scriptsize Yorkshire \& Humber & \scriptsize \textbf{0.99} & \scriptsize \textbf{0.89} $\dagger$ & \scriptsize \textbf{0.98} & \scriptsize \textbf{0.99} & \scriptsize \textbf{0.96} $\dagger$ \\ \hline
\end{tabular}
\begin{tablenotes}
\footnotesize
\item Notes: The table reports the ratios of realised MSFEs of ADMA relative to eDMA, DMA\textsuperscript{0.99}, DMA\textsuperscript{0.95}, BMA and TVP. Figures highlighted in bold indicate that the realised MSFE of the ADMA is less than that of the alternative model. A $\dagger$ indicates rejection of the null hypothesis of the~\cite{ClarkWest_2007}  test at the 5\% significance level.
\end{tablenotes}
\end{table*}

\begin{table*}[h!]
\caption{Forecasting Performance of Alternative Forecasting Strategies Relative to ADMA}
\label{DMAvsADMA}
\centering
\begin{tabular}{llllll}
  \hline\hline
  \tiny (1) & \tiny (2) & \tiny (3) & \tiny (4) & \tiny (5) & \tiny (6) \\
\scriptsize Region & \scriptsize eDMA & \scriptsize DMA{$_{0.99}$} & \scriptsize DMA{$_{0.95}$} & \scriptsize BMA & \scriptsize DLM{$_{0.99}$} \\ \hline

\scriptsize East Anglia & \scriptsize 1.29  & \scriptsize 1.28 & \scriptsize 1.54 & \scriptsize 1.65 & \scriptsize 1.91 \\ 

\scriptsize East Midlands & \scriptsize \textbf{0.98} $\dagger$ & \scriptsize 1.04 & \scriptsize 1.03 & \scriptsize \textbf{0.98} $\dagger$ & \scriptsize 1.93 \\ 

\scriptsize Greater London & \scriptsize 1.06  & \scriptsize 1.07 & \scriptsize 1.04 & \scriptsize \textbf{0.99} & \scriptsize 1.02\\ 

\scriptsize Northern Ireland & \scriptsize 1.08 & \scriptsize 1.03 & \scriptsize 1.12 & \scriptsize 1.14 & \scriptsize 1.18 \\

\scriptsize North & \scriptsize 1.05 & \scriptsize 1.09 & \scriptsize 1.03 & \scriptsize \textbf{0.97} & \scriptsize 1.24 \\

\scriptsize North West & \scriptsize \textbf{0.99} & \scriptsize 1.10  & \scriptsize 1.01 & \scriptsize 1.01 & \scriptsize 1.22 \\

\scriptsize Outer Metropolitan & \scriptsize 1.08 & \scriptsize 1.19 & \scriptsize 1.01 & \scriptsize 1.00 & \scriptsize 1.33 \\

\scriptsize Outer South East & \scriptsize 1.10 & \scriptsize 1.16 & \scriptsize 1.14 & \scriptsize 1.16 & \scriptsize 1.55 \\

\scriptsize Scotland & \scriptsize \textbf{0.96} $\dagger$  & \scriptsize 1.05 & \scriptsize 1.09 & \scriptsize 1.02 & \scriptsize 1.12 \\

\scriptsize South West & \scriptsize \textbf{0.98} $\dagger$ & \scriptsize \textbf{0.98} $\dagger$  & \scriptsize 1.15 & \scriptsize 1.14 & \scriptsize 1.76  \\

\scriptsize West Midlands & \scriptsize 1.14 & \scriptsize 1.12 & \scriptsize 1.12 & \scriptsize 1.11 & \scriptsize 1.84 \\

\scriptsize Wales & \scriptsize \textbf{0.98} & \scriptsize 1.12 & \scriptsize 1.00 & \scriptsize \textbf{0.98} & \scriptsize 1.08 \\

\scriptsize Yorkshire \& Humber & \scriptsize 1.01 & \scriptsize 1.12 & \scriptsize 1.02 & \scriptsize 1.01 & \scriptsize 1.04 \\ \hline
\end{tabular}
\begin{tablenotes}
\footnotesize
\item Notes: The table reports the ratios of realised MSFEs of eDMA, DMA$_{0.99}$, DMA$_{0.95}$, BMA and DLM$_{0.99}$ relative to ADMA. Figures highlighted in bold indicate that the realised MSFE of the competing forecasting strategy is less than that of ADMA. A $\dagger$ indicates rejection of the null hypothesis of the~\cite{ClarkWest_2007}  test at the 5\% significance level.
\end{tablenotes}
\end{table*}

It is evident from Tables~\ref{ForecastAccuracy}, \ref{ADMAvsDMA} and \ref{DMAvsADMA} that
ADMA performs better than all other methods. First, it is the only method that achieves a statistically significant improvement over the benchmark in all regional markets and, second, it produces on average the most accurate forecasts with a mean MSFE 15\% lower than  the AR(1). 
The second best forecasting method is eDMA. This method generates significantly more accurate forecasts than the AR(1) model in nine regional markets and achieves a 10\% average improvement in MSFE relative to the
benchmark. 
A comparison of ADMA and eDMA suggests that ADMA is more accurate in eight of
the 13 regions, with this improvement being statistically significant in six
cases.
In contrast, eDMA performs significantly better than ADMA in only three regions.
DMA with fixed forgetting also outperforms the AR(1) benchmark in the majority
of cases but its performance depends critically on the choice of $\lambda$ and
$\alpha$, and no choice appears to be uniformly better.
ADMA generates more accurate forecasts than DMA$_{0.99}$ and DMA$_{0.95}$
in all regional property markets but one. This forecast improvement is statistically significant in 11 regions for DMA$_{0.99}$ and in nine regions for DMA$_{0.95}$.
BMA achieves a lower MSFE than ADMA in four markets but in all cases the
difference is very small, and only once it is found to be statistically
significant. In contrast, ADMA achieves a significant improvement over BMA in
five regional markets.
Finally, the performance of DLM$_{0.99}$ is uniformly worse than that of ADMA, and this model outperforms the benchmark only in six regional markets.
This outcome is consistent with~\cite{KoopKorobilis_2012}
and~\cite{BorkMoller_2015}, who argue that the use of a large number of
explanatory variables can cause model over-fitting which leads to
inaccurate predictions.

\subsection{Best House Price Predictors Over Time and Across Regions}

Having discussed forecast accuracy, we employ the estimated ADMA weights, $w_{k,t+1}$, to identify important variables for predicting future property price movements, and to
investigate how the \textit{best} house predictors vary over time and across
regional markets. Following~\cite{KoopKorobilis_2012}, for each predictor in
our dataset, we scan through the set of DLMs and select those which contain the
variable under consideration in their specification. The probability that ADMA
assigns to this subset of models, called the \textit{posterior inclusion
probability}, reflects the importance of the variable in forecasting. 

\begin{figure}[t] \centering
\includegraphics[width=1\linewidth]{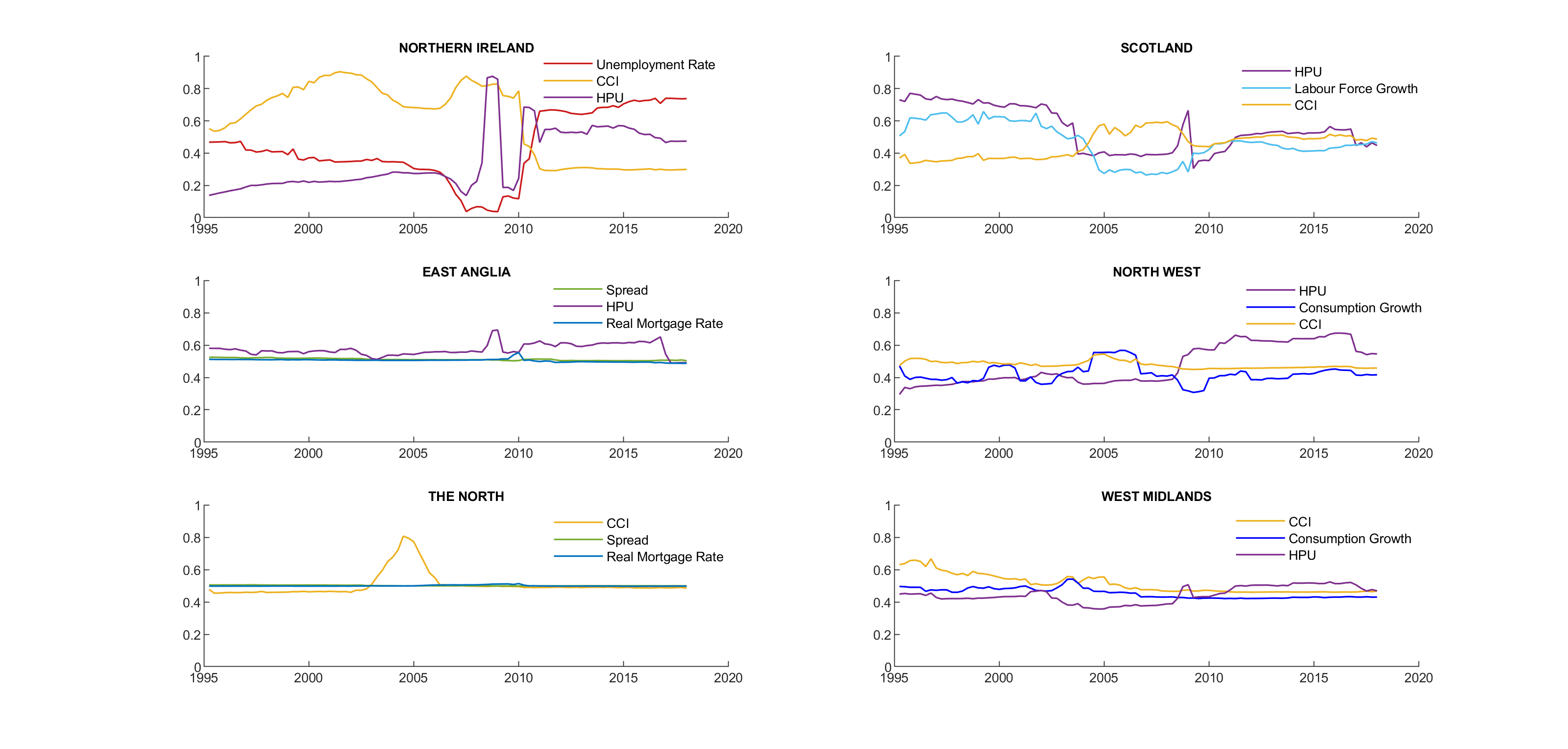}
\caption{Posterior inclusion probabilities for the three most important predictors in the three most volatile UK regional markets (Northern Ireland, East Anglia, the North) and the three most stable markets (Scotland, North West, West Midlands)}\label{REGIONS}
\end{figure}

Figure~\ref{REGIONS} displays the estimated posterior inclusion probabilities.
For presentation purposes, we report results for the three most important
predictors - classifying a predictor as important on the basis of its ADMA
weights over the entire evaluation period- and focus on the three most volatile
(Northern Ireland, East Anglia, the North) and the three most stable regional
markets (Scotland, North West, West Midlands). Overall, the results in
Figure~\ref{REGIONS} suggest that the best predictors differ over time and
across regions. 

For volatile regions, we observe that in two out of the three regions (Northern
Ireland and the North), the key house price predictor during the recent boom is
CCI. In Northern Ireland, which is the most volatile region in our sample, the
posterior inclusion probability attached to CCI is consistently high throughout
the boom phase but drops in the last part of the sample period. In the North,
the CCI posterior inclusion probability increases from around 40\% in 1995 to
80\% in 2004 and then falls back to its original level. These findings are in
line with the widely held view that changes in credit conditions were at the
heart of the house price surge prior to the financial crisis. 

On the other hand, house price uncertainty plays an important role in
predicting future house price movements in volatile markets ahead of the house
price collapse of 2008:Q3. The probability of including HPU in the forecasting
models of Northern Ireland and East Anglia rises to around 90\% and 70\%,
respectively, in 2008:Q3 and then drops following the downturn in property
prices. The mortgage rate and spread are important predictors of house price
inflation in East Anglia and the North, though their ADMA weights are only
marginally above 0.5. With regard to the other predictive variables, we note
that these are important in some volatile regions, but not in others. Perhaps
the most striking example is the unemployment rate. For Northern Ireland, this
variable is one of the key determinants of future property price movements in
the aftermath of the house price collapse, with an ADMA weight of around 0.7
from 20011:Q1 until the end of the sample. On the contrary, for East Anglia and
the North, the probability of including the unemployment rate in the predictive
model is never above 0.5.

Moving on to the stable housing markets of Scotland, North West and West
Midlands, we notice that credit availability and house price uncertainty are
again included in the set of important predictors. In all three regions, HPU
becomes the best predictor ahead of the house price collapse of 2008:Q3 and
during the bust phase. While, CCI is the key determinant of property price
inflation at the start of the boom phase, from the first quarter of 2004 until
the end of 2005. Similarly to volatile property markets, the remaining
predictors show mixed predictive ability. 
Overall the results of the empirical application suggest that allowing for
structural instability and regional heterogeneity is crucial for forecasting UK
house prices.

\section{Conclusions}\label{sec:Conc}

Dynamic model averaging (DMA) is gaining increasing attention in macroeconomic
time series forecasting due to its ability to accommodate time-variation in
both the parameters as well as the specification of the optimal forecasting
model.  In this paper we introduced a novel adaptive methodology for DMA
which aims to overcome limitations of existing DMA specifications with respect
to both the sequential estimation of the optimal forgetting factor for each
dynamic linear model (DLM), as well as the model averaging process.
Motivated by work in adaptive filtering, we proposed to optimise the forgetting
factor of each DLM through a state-of-the-art stochastic gradient descent
algorithm. Our simulation study illustrated that this approach can effectively
approximate the optimal forgetting factor 
under different types of change in the data generating process, including cases
in which the speed or type of change is variable over time. 
A further advantage of our approach is that it is computationally less
demanding compared to competing DMA specifications that sequentially update the
DLM forgetting factor by considering a grid of values for this parameter.
Our adaptive methodology also involves a parameter-free forecast aggregation
algorithm from the literature on prediction with expert advice. This allows us
to obtain finite-time performance guarantees about the forecast accuracy of the
DMA forecast. To the best of our knowledge no other DMA specification has this
property.

We conducted an in-depth empirical evaluation of the proposed methodology on
the task of forecasting UK regional house prices.  Our results indicate that
the adaptive DMA produces overall more accurate forecasts than competing DMA
specifications.  They also reveal that no single predictor is consistently
chosen as the key determinant of future property price movements. Credit
availability was found to be an important predictor of house price inflation
for several regional markets during the boom phase of the 2000s, while house
price uncertainty appeared to play an important role in predicting house price
movements on the eve of the price collapse of 2008:Q3.

\bibliographystyle{chicago}



\clearpage
\newpage
\begin{appendices}

\renewcommand{\theequation}{\thesection.\arabic{equation}}
\setcounter{equation}{0}

\section{Adaptive Forgetting DLM}\label{app:Deriv}

In this appendix we derive all the derivatives necessary to compute $\dJdl{t+1}$.
For completeness, we repeat the chain rule equation,
\begin{equation}\label{eq:ChainRule2}
\dJdl{t+1} = \frac{\partial J_{t+1}}{\partial \est_t} \dTdl{t}.
\end{equation}
First, we introduce two equations that follow directly from the definition of DLM, and which we will need in the following derivations.
Replacing the definition of the adaptive coefficient vector~$A_t$, Eq.~\eqref{eq:KG},
in Eq.~\eqref{eq:Koop9} yields the following equivalent expression for the estimator of the conditional covariance of $\theta_t$,
\begin{align}\label{eq:AltC}
C_{t} & = \lambda^{-1} C_{t-1} - A_{t} x_{t}^\top \lambda^{-1} C_{t-1}.
\end{align}
Combining the above with Eq.~\eqref{eq:KG} yields,
\begin{align}\label{eq:KG2}
A_{t} & = C_{t} x_{t} S_{t-1}^{-1}.
\end{align}

We now proceed to the derivation of $\dJdl{t}$. Obtaining an expression for the first derivative in Eq.~\eqref{eq:ChainRule} is straightforward,
\begin{equation}
\frac{\partial J_{t+1}}{\partial \est_t}  = - \he_{t+1} x_{t+1}^\top.
\end{equation}
The derivative of $\est_t$ with respect to the forgetting factor, $\dTdl{t}$, is obtained by differentiating Eq.~\eqref{eq:Koop8} with respect to $\lambda$,
\begin{align}\label{eq:dTdl}
\dTdl{t} & = \nabla_{\lambda} \left\{  \est_{t-1} + A_{t} \he_{t} \right\}  \nonumber \\
& = \nabla_{\lambda} \left\{ \est_{t-1} + C_{t} x_{t} S_{t-1}^{-1} \he_{t} \right\} \nonumber \\
& = \dTdl{t-1} + \dCdl{t} x_{t} S_{t-1}^{-1} \he_{t} - C_{t} x_{t} S_{t-1}^{-2} \dSdl{t-1}\he_{t} - C_{t} x_{t} S_{t-1}^{-1} x_t^\top \dTdl{t-1} \nonumber \\
& = (I - C_{t} x_{t} S_{t-1}^{-1} x_t^\top) \dTdl{t-1} + \dCdl{t} x_{t} S_{t-1}^{-1} \he_{t} - C_{t} x_{t} S_{t-1}^{-2} \dSdl{t-1} \he_{t} \nonumber \\
& = (I - A_{t-1}x_t^\top) \dTdl{t-1} + \frac{\he_t}{S_{t-1}} \left(\dCdl{t} x_{t} - A_{t-1} \dSdl{t-1} \right),
\end{align}
where $\dCdl{t}$ is the derivative of $C_t$, the estimator of the covariance matrix of $\theta_t$, with respect to $\lambda$, and $\dSdl{t-1}$ is the derivative of the point estimate of the observational variance at time $t-1$, $S_{t-1}$ with respect to~$\lambda$.
In the first and last steps of the above derivation we used Eq.~\eqref{eq:KG2}.

We next need to derive expressions for $\dCdl{t}$ and $\dSdl{t}$.
To this end we will need the derivative of the one-step-ahead predictive variance $Q_t$ with respect to
$\lambda$, which we denote as $\dQdl{t}$. By differentiating Eq.~\eqref{eq:Q} we get,
\begin{align*}
\dQdl{t} = \lambda^{-1} x_t^\top \dCdl{t-1} x_t - \lambda^{-2} x_t^\top C_{t-1} x_t + \dSdl{t-1}.
\end{align*}
We are now able to get an expression for $\dSdl{t}$ by differentiating Eq.~\eqref{eq:PW_CN},
\begin{align}
\dSdl{t} = &\dSdl{t-1} + \frac{1}{n_t}\dSdl{t-1}  \left( \frac{\he_t^2}{Q_{t}} - 1 \right)  \nonumber \\
& - \frac{S_{t-1}}{n_t} \left( \frac{2 \he_t x_t^\top \dTdl{t-1} }{Q_t}  +\frac{ \he_t^2 }{Q_t^2}\dQdl{t} \right).
\end{align}

\noindent
Next we compute the derivative $\dCdl{t}$ by differentiating Eq.~\eqref{eq:AltC},
\begin{align}\label{eq:dCdl}
\dCdl{t} = & \lambda^{-1} \dCdl{t-1} -\lambda^{-2} C_{t-1} - \dAdl{t} x_t^\top \lambda^{-1} C_{t-1} - \nonumber \\
& -  A_t x_t^\top  \lambda^{-1} \dCdl{t-1} + A_t x_t^\top \lambda^{-2} C_{t-1},
\end{align}
where $\dAdl{t}$ denotes the derivative of the adaptive coefficient vector at time $t$ with respect
to $\lambda$. This derivative is obtained by differentiating Eq.~\eqref{eq:KG} after expanding
the term $Q_t$ that appears in the denominator by its definition in Eq.~\eqref{eq:Q}.
In particular,

\begin{align} \label{eq:dAdl}
\dAdl{t} & =  \frac{ \lambda^{-1} \dCdl{t-1} x_t}{Q_t} - \frac{\lambda^{-2} C_{t-1} x_{t} }{ Q_t } 
- \frac{\lambda^{-1}C_{t-1} x_{t} \dQdl{t}}{Q_t^2} \nonumber\\
& =  \frac{1}{Q_t} \Bigg\{ \lambda^{-1} \dCdl{t-1} x_t - \lambda^{-2} C_{t-1} x_{t} \nonumber \\
&\quad - A_t \left(\lambda^{-1} x_t^\top \dCdl{t-1} x_t - \lambda^{-2} x_t^\top C_{t-1} x_t + \dSdl{t-1}\right) \Bigg\} \nonumber \\
& = \frac{1}{\lambda Q_t} \dCdl{t-1} x_t - A_t \left( \frac{1}{\lambda} + \frac{1}{Q_t} \dQdl{t} \right).
\end{align}
Substituting the above expression for $\dAdl{t}$ into the derivative $\dCdl{t}$ in Eq.~\eqref{eq:dCdl},

\begin{align} 
\dCdl{t} & = \lambda^{-1} \dCdl{t-1} -\lambda^{-2} C_t - \left(\lambda^{-1} \dCdl{t-1} x_t - \lambda^{-2} C_{t-1} x_{t} \right) \frac{\lambda^{-1} x_t^\top C_{t-1}}{ Q_t } \nonumber \\
&\quad + A_t \left(\lambda^{-1} x_t^\top \dCdl{t-1} x_t - \lambda^{-2} x_t^\top C_{t-1} x_t + \dSdl{t-1}\right) 
\frac{\lambda^{-1} x_t^\top C_{t-1}}{ Q_t } \nonumber \\
&\quad -  A_t x_t^\top  \lambda^{-1} \dCdl{t-1} + A_t x_t^\top \lambda^{-2} C_{t-1} \nonumber \\
& =  \lambda^{-1} \dCdl{t-1} -\lambda^{-2} C_t - \left(\lambda^{-1} \dCdl{t-1} x_t - \lambda^{-2} C_{t-1} x_{t} \right) A_t^\top \nonumber \\
&\quad + A_t \left(\lambda^{-1} x_t^\top \dCdl{t-1} x_t - \lambda^{-2} x_t^\top C_{t-1} x_t + \dSdl{t-1}\right) A_t^\top \nonumber \\
&\quad -  A_t x_t^\top  \lambda^{-1} \dCdl{t-1} + A_t x_t^\top \lambda^{-2} C_{t-1}.
\end{align}

\noindent
Substituting Eqs.~\eqref{eq:AltC} and~\eqref{eq:KG2} in the above equation and collecting terms yields,
\begin{align}\label{eq:dCdl2}
\dCdl{t} & = \lambda^{-1} \left(I - A_t x_t^\top \right) \dCdl{t-1}\left(I - x_t A_t^\top \right)  \nonumber \\
&\quad + A_t \dSdl{t-1} A_t^\top + \lambda^{-1} C_{t} + \lambda^{-1} C_t x_t A_t^\top \nonumber \\
& = \lambda^{-1} \left(I - A_t x_t^\top \right) \dCdl{t-1}\left(I - x_t A_t^\top \right)  \nonumber \\
&\quad + A_t \dSdl{t-1} A_t^\top + \lambda^{-1} C_{t} + \lambda^{-1} A_t S_{t-1} A_t^\top.
\end{align}

The above derivation illustrates that the derivatives of all the involved
quantities can be expressed as functions of quantities involved in the update equations of the DLM, and $\dCdl{t}$ and $\dSdl{t}$, which can be estimated recursively.
We are now in position to state the adaptive forgetting DLM algorithm.
Algorithm~\ref{alg:adlm} contains a detailed description,
including the ADAM stochastic gradient descent algorithm which is used to
update~$\lambda$ in consecutive iterations. 
Although the algorithm is expressed in terms of quantities at time $t$ and $t-1$, the computations are ordered in a manner that allows us to perform all the updates by retaining a single copy of each quantity. In other words, there is no need to retain lagged values for any of the variables involved. Also note that the update equations for $\dTdl{t}$ and $\dCdl{t}$ in Algorithm~\ref{alg:adlm} are obtained by differentiating the update equations for~$\est_t$ and~$C_t$, respectively.

The first two inputs, $\{x_t\}_{t=1}^T$ and $\{y_t\}_{t=1}^T$, correspond to the time series of the covariates and response, respectively.
The third input parameter, $g$, determines the variance of the prior distribution for $\theta$,
see Eq.~\eqref{eq:C0}. We follow~\cite{CataniaNonejad_2018} and use as
default value~$g=100$.
The next three inputs $\lambda_0, \lambda^+, \lambda_{-}$, correspond to the
initial, the maximum and the minimum value of the forgetting factor. The default
values for these parameters are 0.99, 0.999 and 0.9, respecively
It is worth noting that bounds on the upper and lower values of the forgetting
factor are employed by all algorithms that perform online tuning of this
parameter~\citep{Haykin_2002,PavlidisTAH2011,Anagn12}.
Finally, the last three parameters are employed by the ADAM stochastic gradient descent algorithm.
The first, $\gamma$, determines the maximum step-size, while $\beta_1$ and $\beta_2$ control
the exponential decay rates for the moving average estimates
of the mean $m_t/(1-\beta_1^t)$ and variance $v_t/(1-\beta_2^t)$ of the estimated gradient $\dTdl{}$ over time. 
\begin{algorithm}[H]
\caption{Adaptive Forgetting DLM}\label{alg:adlm}
\begin{algorithmic}[1]
\REQUIRE $\{x_t\}_{t=1}^T, \{y_t\}_{t=1}^T,g=100,\lambda_1=0.99, \lambda^+ = 0.999, \lambda_{-} = 0.9, 
	\gamma =5 \cdot 10^{-3},\beta_1=0.8, \beta_2=0.8$ 

\STATE $\est_{0} = 0$

\STATE $\he_1 = y_1 $

\STATE $C_1 = g I$
\STATE $\dCdl{1} = 0$

\STATE $Q_1 = x_1^\top C_1 x_1$

\STATE $A_1 = C_1 x_1/Q_1$

\STATE $\est_1 = \est_0 + \he_1 A_1$
\STATE $\dTdl{1}= 0$

\STATE $S_1 = \frac{1}{2}\left(y_1^2 + \he_1^2/Q_1 \right)$	
\STATE $\dSdl{1} = 0$	

\STATE $n_1 = 2$

\STATE $m_1 = 0$ \COMMENT{Initialise ADAM Parameters} 
\STATE $v_1 = 0$

\FOR{ $t = 2,\ldots, T$}

\STATE $n_t = n_{t-1} + 1$

\STATE $\hat{y}_t = x_t^\top \est_{t-1}$ \COMMENT{Prediction}

\STATE $Q_t = \lambda_{t-1}^{-1} x_{t}^\top C_{t-1} x_{t} + S_{t-1}$ \COMMENT{Prediction variance}

\STATE $\dQdl{t} = \lambda_{t-1}^{-1} x_t^\top \dCdl{t-1} x_t - \lambda_{t-1}^{-2} x_t^\top C_{t-1} x_t + \dSdl{t-1}$

\STATE $\he_t = y_t - \hat{y}_t$ \COMMENT{Prediction error}

\STATE $\dJdl{t} = -\he_{t} x_{t}^\top \dTdl{t-1}$ \COMMENT{Derivarive of forecast error w.r.t $\lambda$}

\STATE $A_{t} =  \lambda_{t-1}^{-1} C_{t-1} x_{t}/Q_{t}$ \COMMENT{Adaptive coefficient vector }

\STATE $\dAdl{t} = \lambda_{t-1}^{-1} Q_t^{-1} \dCdl{t-1} x_t - A_t \left( \lambda_{t-1}^{-1} + \dQdl{t} Q_t^{-1}\right)$

\STATE $\dSdl{t} = \dSdl{t-1} + \frac{1}{n_t Q_t} \left(\dSdl{t-1} \left( \he_t^2 - Q_t \right) - S_{t-1}
\left(2\he_t x_t^\top \dTdl{t-1} + \frac{\he_t^2}{Q_t} \dQdl{t} \right) \right)$

\STATE $S_t = S_{t-1} + n_t^{-1} S_{t-1} (\he^2 Q_t^{-1} -1)$ \COMMENT{Observational noise}

\STATE $\est_t = \est_{t-1} + A_t \he_t$ \COMMENT{Posterior coefficient estimate}

\STATE $\dTdl{t} = \dTdl{t-1} + \he_t \dAdl{t} - x_t^\top \dTdl{t-1} A_t$

\STATE $\dCdl{t} = (I - A_t x_t^\top) \lambda_{t-1}^{-1} \dCdl{t-1} - 
	(I + \lambda_{t-1} \dAdl{t}x_t^\top - A_t x_t^\top)\lambda_{t-1}^{-2} C_{t-1}$

\STATE $C_{t} = \lambda_{t-1}^{-1} C_{t-1} - A_t A_t^\top Q_t$ \COMMENT{Coefficient covariance matrix}
\algstore{testcont}
\end{algorithmic}
\end{algorithm}

\begin{algorithm}[H]                     
\begin{algorithmic}[1]                 
\algrestore{testcont}
\STATE \tt{\# ADAM Stochastic Gradient Descent ($\epsilon = 10^{-8}$)}

\STATE $m_t = b_1 m_{t-1} + (1-b_1) \dJdl{t}$

\STATE $v_t = b_2 v_{t-1} + (1-b_2) \left(\dJdl{t}\right)^2$

\STATE $\lambda_t = \left[ \lambda_{t-1} - \gamma m_t \left((1-\beta_1^t) \left(\sqrt{v_t/(1-\beta_2^t)} + \epsilon\right) \right)^{-1} \right]_{\lambda^{-}}^{\lambda^{+}}$ 

\ENDFOR

\end{algorithmic}
\end{algorithm}

\newpage
\section{Variable Definitions and Data Sources} \label{appendix:a}

\textbf{House Prices} Real regional house price index (all houses, seasonally adjusted). Source: Nationwide.
\newline
\textbf{Income} Real average total household's weekly expenditure. Source: Family Expenditure Survey (FES). 
\newline
\textbf{Price-to-Income Ratio} Quarterly changes in the log of the ratio of house prices to income.
\newline
\textbf{Income Growth}~Annualised quarterly changes in the log of real income.
\newline
\textbf{Labour Force Growth}~Annualized quarterly changes in the log of the sum of unemployed and employed people. Source: Labour Force Survey (LFS).
\newline
\textbf{Unemployment Rate}~Quarterly changes in the ratio of unemployed people to the labour force times 100. Source: LFS.  
\newline
\textbf{Real Mortgage Rate}~Quarterly changes in the real mortgage rate of building societies, adjusted for the cost of mortgage tax relief as in \citet{Muellbauer_2006}. Sources: OECD Main Economic Indicators and HM Revenue \& Customs.
\newline
\textbf{Spread}~Difference between the 10-year government bond yield and the rate of discount on 3-month treasury bills. Sources: Saint Louis FRED Economic Data and the Bank of England.
\newline
\textbf{Industrial Production Growth}~Annualised quarterly changes in the log of total industrial production of all industries (seasonally adjusted). Source: Office for National Statistics (ONS).
\newline
\textbf{Real Consumption Growth}~Annualised quarterly changes in the log of real final consumption expenditure of households and non-profit institutions serving households (seasonally adjusted, millions of UK sterling pounds). Source: ONS. 
\newline
\textbf{Housing Starts}~Log of the number of all permanent dwellings started in the UK. Source: Department for Communities and Local Government.
\newline
\textbf{Index of Credit Conditions}~Designed as a linear spline function, this index is estimated using a two-equation system of secured and unsecured lending. For details about the methodology and sources of the data used in the estimation please refer to the supplementary Appendix to \cite{PavlidisPPY2017}. 
\newline
\textbf{House Price Uncertainty Index}~Constructed using the methodology outlined in \citet{BakerScott_2016} to proxy for economic policy uncertainty. The HPU is an index of search results from five large newspapers in the UK: The Guardian, The Independent, The Times, Financial Times and Daily Mail. We use LexisNexis digital archives of these newspapers to obtain a quarterly count of articles that contain the following three terms: \bsq{uncertainty} or \bsq{uncertain}; \bsq{housing} or \bsq{house prices} or \bsq{real estate}; and one of the following: \bsq{policy}, \bsq{regulation}, \bsq{Bank of England}, \bsq{mortgage}, \bsq{interest rate}, \bsq{stamp-duty}, \bsq{tax}, \bsq{bubble} or \bsq{buy-to-let} (including variants like \bsq{uncertainties}, \bsq{housing market} or \bsq{regulatory}). To meet the search criteria an article must contain terms in all three categories. The resulting series of search counts is then scaled by the total number of articles in the given newspaper and in the given quarter. Finally, to obtain the HPU index, we average across the five newspapers by quarter and normalise the index to a mean of 100.

\end{appendices}

\bibliographystyle{chicago}

\end{document}